\documentclass[runningheads, a4paper]{llncs} 
\usepackage{graphicx,amssymb,amsmath,comment, color, url}  % Salim added on July 03, 2011
\newtheorem{obs}{Observation}

\begin{document}

\mainmatter

\title{Improved upper and lower bounds for the point placement problem}

%\titlerunning{Improved upper and lower bounds for the point placement problem}
\titlerunning{Point placement problem}

%\author{Md. Shafiul Alam
%\and Asish Mukhopadhyay \thanks{Research supported by an NSERC discovery grant awarded to this author.}}

\author{Md. Shafiul Alam 
\and Asish Mukhopadhyay \thanks{Research supported by an NSERC discovery grant awarded to this author.}}

%\authorrunning{Md. Shafiul Alam and Asish Mukhopadhyay}
\authorrunning{Md. Shafiul Alam and Asish Mukhopadhyay}
\institute{School of Computer Science, University of Windsor,\\401 Sunset Avenue, Windsor, ON, N9B 3P4, Canada\\ 
%\mail {\{alam9, asishm\}@uwindsor.ca}}\\
%\url{http://www.springer.com/lncs}
}


%\toctitle{Lecture Notes in Computer Science}
%\tocauthor{Authors' Instructions}

\maketitle

\begin{abstract}
The point placement problem is to determine the positions of a set of $n$ distinct points, $P = \{p_1, p_2, p_3, \ldots, p_n\}$, on a line uniquely, up to translation and reflection, from the fewest possible distance queries between pairs of points. Each distance query corresponds to an edge in a graph, called point placement graph ($ppg$), whose vertex set is $P$. The uniqueness requirement of the placement translates to line rigidity of the $ppg$. In this paper we show how to construct in 2 rounds a line rigid point placement graph of size $9n/7 + O(1)$. This improves the result reported in ~\cite{DBLP:conf/wabi/ChinLSY07} for 5-cycles. We also improve the lower bound on 2-round algorithms from $17n/16$ ~\cite{DBLP:conf/wabi/ChinLSY07} to $9n/8$.
\end{abstract}
%\vspace{.1in}

%\noindent {\em Keywords:} point placement problem, learning by queries, rigid graphs

\section{Introduction}  
  
%\subsection{The problem}  
 
%OK 
Let $P = \{p_1, p_2, . . . , p_n\}$ be a set of $n$ distinct points   
on a line $L$. In this paper, we address the problem of determining a unique placement   
(up to translation and reflection) of the $p_i$'s   
on $L$, by querying distances   
between some pairs of points $p_i$ and $p_j$, $1 \leq i, j \leq n$.
The resulting queries can be represented by   
a \emph{point placement graph} ($ppg$, for short), $G = (P, E)$, where each edge $e$ in $E$   
joins a pair of points $p_i$ and $p_j$   
in $P$ if the distance between these two points on $L$ is known and the length of $e$, $|e|$, is the   
distance between the corresponding pair of points. (Note the dual use of $p_i$ to denote a   
point on $L$ as well as a vertex of $G$.)
We will say that $G$ is \emph{line rigid} or just \emph{rigid} when there is a   
unique placement for $P$. Thus, the original problem reduces to the construction  of a line rigid $ppg$, $G$.

Early research on this problem was reported in~\cite{Redstone:2000:ASP:648257.752909,journals/dam/Mumey00}.   
In this paper, our first principal reference is \cite{DBLP:journals/dam/Damaschke03}, where it was shown that   
%the jewel (Fig.~\ref{fourfourjewel:label}) 
 jewel and $K_{2,3}$ are both line rigid, as also   
how to build large rigid graphs of density 8/5 (this is an asymptotic measure of the   
number of edges per point as the number of points go to infinity) out of the jewel. In a subsequent   
paper, Damaschke \cite{DBLP:journals/dam/Damaschke06} proposed a randomized 2-round strategy that needs   
$(1+o(1))n$ distance queries   
with high probability and also showed that this is not possible with 2-round deterministic strategies. He also reported the following result:

\begin{obs}{\label{obs:one}}   
%{\bf Observation.} At most two equal length edges that are collinear with $L$ can be incident to a point $p$ on $L$. \\  
At most two equal length edges that are collinear with a line $L$ can be incident to a point $p$ on $L$.%%%%%%%%%%%%%%%%%%%%%%1. Alam Jan. 18, 2011 
\end{obs}
 
Our second principal reference is the   
work of~\cite{DBLP:conf/wabi/ChinLSY07} who improved many of the results of \cite{DBLP:journals/dam/Damaschke03}. Their   
principal contributions are the 3-round construction of rigid graphs of density 5/4 from 6-cycles and a lower bound on the   
number of queries necessary in any 2-round algorithm. They also introduced the idea of a layer graph which is useful in finding the conditions for rigidity of a \textit{ppg} and proved the following result about it:

\begin{comment}
, and is defined as follows:

\begin{definition} 
We first choose two orthogonal directions  {\bf x} and {\bf y} (actually, any 2 non-parallel directions will do). A graph $G$ admits a layer graph drawing if the following 4 properties are satisfied:   
   
\begin{enumerate}     
\item[P1] Each edge $e$ of $G$ is  parallel to one of the two orthogonal directions, {\bf x} and {\bf y}.     
         
\item[P2] The length of an edge $e$ is the distance between the corresponding points on $L$.      
         
\item[P3] Not all edges are along the same direction (thus a layer graph has a two-dimensional extent).     
         
\item[P4] When the layer graph is folded onto a line, by a rotation either to the left or to the right about an edge of the layer     
graph lying on this line, no two vertices coincide.     
\end{enumerate}   
\end{definition}

Chin \textit{et al.}~\cite{DBLP:conf/wabi/ChinLSY07} proved the following result about a layer graph:     
\end{comment}

\begin{theorem}\label{lr=lgtheorem:label}     
A $ppg$ G is line rigid iff it cannot be drawn as a layer graph.     
\end{theorem}

In ~\cite{am-anewalg} we proposed a 2-round algorithm that query $4n/3 + O(1)$ edges to construct line rigid $ppg$ on $n$ points using 6:6 jewels  as the basic components. In this paper, we propose a 2-round algorithm that queries $9n/7 + O(1)$ edges to construct a line rigid $ppg$ on $n$ points, using 3 paths of degree two nodes of length 2 each with a common vertex as the basic component, bettering 
a result of~\cite{DBLP:conf/wabi/ChinLSY07} that uses 5-cycles. More significantly, we improve their lower bound on any 2-round algorithm from $17n/16$ to $9n/8$.

\section{A Two Round Algorithm}

We shall use 3 paths $p_1q_1r_1s$, $p_2q_2r_2s$ and $p_3q_3r_3s$ of degree 2 nodes of length 2 attached to a node $s$ of degree 3 as the basic component for the point placement (Fig.~\ref{3pathfig:label}).
 Other ends $p_1$, $p_2$ and $p_3$ of the 3 paths are made line rigid in the first round. %We want to place the remaining 7 points. 
We shall make the remaining 7 points line rigid in the second round. We find a set of sufficient conditions that make the component line rigid by preventing its drawing as a \emph{layer graph} (Theorem~\ref{lr=lgtheorem:label}).  %~\cite{DBLP:conf/wabi/ChinLSY07}. 
We shall call the component as 3-path and the corresponding \textit{ppg} as 3-path \textit{ppg}. To find the rigidity conditions we consider $(p_1, q_1, r_1, s, r_2, q_2, p_2)$ as a 7-cycle. We shall find conditions that will make the 7-cycle line rigid. Then $s$ will be unambiguous. Also $p_1$, $p_2$ and $p_3$ are fixed in the first round. Consequently, the distance between  $p_3$ and $s$ will be fixed. So, we can consider $(p_3, q_3, r_3, s)$ as a 4-cycle. We shall find the condition for rigidity of this 4-cycle. Then the union of these two sets of conditions will comprise the set of rigidity conditions for the whole 3-path \textit{ppg}.

%\begin{comment}
\begin{figure}[!h] 
\centering 
\includegraphics[scale=0.7]{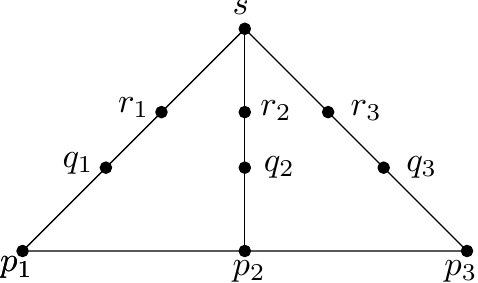} 
\caption{The 3-path basic component.}
\label{3pathfig:label} 
\end{figure}
%\end{comment}

\begin{comment}
\begin{figure}[ht]     
\begin{minipage}[b]{0.45\linewidth}     
\centering     
\includegraphics[scale=0.7]{3path.pdf}     
\caption{The 3-path basic component.}     
\label{3pathfig:label}     
\end{minipage}     
\hspace{0.2cm}     
\begin{minipage}[b]{0.45\linewidth}     
\centering     
\includegraphics[scale=0.7]{configa.pdf}     
\caption{The layer graph representation of the 7-cycle $(p_1, q_1, r_1, s, r_2, q_2, p_2)$ corresponding to the condition $|p_1p_2| \ne |q_2r_2|$.}
\label{configa:label}     
\end{minipage}     
\end{figure}   
\end{comment}

We shall attach all the basic components to triplets of points among some constant number of line rigid (in the first round) points $p_i$. Then for each component there will be extra 7 points and 9 edges. Thus, the density will be $O(9/7)$. %To find the rigidity conditions we consider $(p_1, q_1, r_1, s, r_2, q_2, p_2)$ as a 7-cycle.

We shall not query the lengths of the edges $q_1r_1$, $q_2r_2$ and $q_3r_3$ in the first round. We shall query them in the second round. So, we shall find a set of sufficient conditions for rigidity for the basic component that does not involve these edges. Then we can satisfy all the rigidity conditions irrespective of the lengths of these edges which will be reported in the second round.

The layer graphs of the 7-cycle $(p_1, q_1, r_1, s, r_2, q_2, p_2)$ can be grouped into 6 groups based on the number of edges on each side (Fig.~\ref{layergraph7cycle:label}). When different configurations of the chain $p_3q_3r_3s$ are attached to them, the total number of layer graphs for the 3-path component becomes 42. From them, by Theorem~\ref{lr=lgtheorem:label}, we get the following 42 conditions for rigidity of the 3-path component:

\begin{figure}[!h] 
\centering 
\includegraphics[scale=0.6]{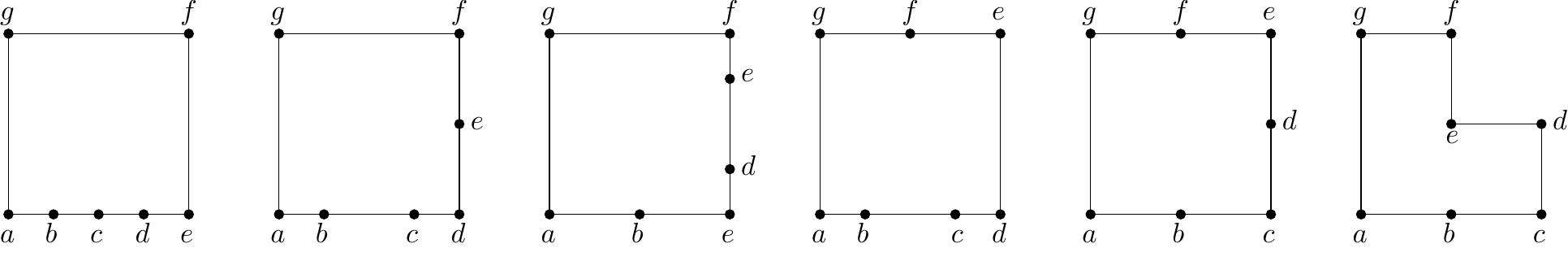} 
\caption{Groups of layer graphs for 7-cycle.}
\label{layergraph7cycle:label} 
\end{figure}

\begin{enumerate}

\item $|p_1p_2| \ne |q_2r_2|$, $|p_1p_2| \ne |q_1r_1|$,
$|p_2q_2| \ne |r_2s|$, $|p_1q_1| \ne |r_1s|$,
$|q_2r_2| \ne |r_1s|$, $|q_1r_1| \ne |r_2s|$,
$|p_1q_1| \ne |p_2q_2|$.

\item $||p_1p_2| \pm |p_2q_2|| \ne |r_2s|$,  
$||p_2q_2| \pm |q_2r_2|| \ne |r_1s|$, 
$||p_1p_2| \pm |p_1q_1| \pm |p_2q_2|| \ne |r_1s|$,  
$|p_1q_1| \ne ||r_1s| \pm |r_2s||$, 
$|p_1p_2| \ne ||q_1r_1| \pm |r_1s||$,  
$||p_1q_1| \pm |q_1r_1|| \ne |p_2q_2|$,
$||p_1q_1| \pm |p_1p_2|| \ne |q_2r_2|$.

\item $||p_1p_2| \pm |p_1q_1|| \ne |r_1s|$,
$||p_1q_1| \pm |q_1r_1|| \ne |r_2s|$, 
$||p_1p_2| \pm |p_1q_1| \pm |p_2q_2|| \ne |r_2s|$,  
$|p_2q_2| \ne ||r_1s| \pm |r_2s||$, 
$|p_1p_2| \ne ||q_2r_2| \pm |r_2s||$,  
$||p_2q_2| \pm |q_2r_2|| \ne |p_2q_2|$,
$||p_2q_2| \pm |p_1p_2|| \ne |q_1r_1|$.

\item $|p_1p_2| \ne |r_2s|$, $|p_1p_2| \ne |r_1s|$, 
$|p_2q_2| \ne |r_1s|$, $|p_1q_1| \ne |r_2s|$, 
$||p_1q_1| \pm |p_2q_2| \pm |p_1p_2|| \ne ||r_1s| \pm |r_2s||$,
$|p_2q_2| \ne |q_1r_1|$, $|p_1q_1| \ne |q_2r_2|$.

\item $|p_2q_2| \ne ||p_1p_2| \pm |r_1s||$, $|p_1q_1| \ne ||p_1p_2| \pm |r_2s||$, 
$|p_1q_1| \ne ||p_1p_2| \pm |r_1s| \pm |r_2s||$, $|p_2q_2| \ne ||p_1p_2| \pm |r_1s| \pm |r_2s||$, 
$|p_1q_1| \ne ||q_2r_2| \pm |r_2s||$, $|p_2q_2| \ne ||q_1r_1| \pm |r_1s||$, 
$|p_1p_2| \ne ||r_1s| \pm |r_2s||$.

\item $||p_1q_1| \pm |q_1r_1|| \ne ||p_2q_2| \pm |r_2s||$, $||p_2q_2| \pm |q_2r_2|| \ne ||p_1q_1| \pm |r_1s||$, 
$|q_1r_1| \ne ||p_2q_2| \pm |r_2s||$, $|q_2r_2| \ne ||p_1q_1| \pm |r_1s||$, 
$|p_2q_2| \ne ||p_1q_1| \pm |r_1s||$, $|p_1q_1| \ne ||p_2q_2| \pm |r_2s||$, 
$|p_2q_2| \ne ||p_1q_1| \pm |r_1s| \pm |r_2s||$.

\end{enumerate}

%Above 42 conditions will make the 7-cycle line rigid. Then $s$ will be unambiguous. Also $p_1$, $p_2$ and $p_3$ are fixed in the first round. Consequently, the distance between  $p_3$ and $s$ will be fixed. So, we can consider $(p_3, q_3, r_3, s)$ as a 4-cycle. We make it line rigid by imposing the condition $|p_3q_3| \ne |r_3s|$. Then the whole component will be line rigid.

Among them 20 conditions involve the edges $q_1r_1$ and $q_2r_2$ of the 7-cycle that we want to avoid in the conditions.
%: jXCj 6= jY Bj, jXCj 6= jABj, jY Cj 6= jABj, jXRj 6= jZQj, jXRj 6= jPQj and jZRj 6= jPQj 
We shall replace each of these conditions by a set of conditions that prevents the 7-cycle from being drawn as the layer graph representation that corresponds to that condition. Collection of all these new conditions and the ones that are not replaced will constitute the rigidity conditions for the 7-cycle. As stated before if the 7-cycle is line rigid then the $(p_3, q_3, r_3, s)$ will be a 4-cycle which can be made line rigid by imposing the condition $|p_3q_3| \ne |r_3s|$~\cite{DBLP:conf/wabi/ChinLSY07}. This condition together with the rigidity conditions for the 7-cycle will constitute the rigidity conditions for the whole component.
%with other conditions that do not involve those 2 edges of unknown length in the first round. 
%From Fig.~\ref{} we see that we can replace the condition jXCj 6= jABj by jXAj 6= jY BjjY Cj. Similarly, we can replace jXRj 6= jPQj by jXPj 6= jZQj  jZRj.

As an example of replacing conditions we shall replace the first condition, viz., $|p_1p_2| \ne |q_2r_2|$, which corresponds to the layer graph representation of the 7-cycle in Fig.~\ref{configa:label}. To replace the condition
we find a set of conditions that prevent the drawing of layer graph of the 7-cycle $(p_1, q_1, r_1, s, r_2, q_2, p_2)$ in the configuration of Fig.~\ref{configa:label}. 
For this we draw all the possible configurations of the layer graph of the whole component with the layer graph of the 7-cycle being in the configuration of Fig.~\ref{configa:label}. 
%\textcolor{red}{xTo replace this condition, we proceed thus. Consider the layer graph drawing of the 5-cycle $XABYC$ in which this condition is violated (Fig.~\ref{fivecyclefig:label}e). We consider all possible layer graphs of the 5:5 jewel in which the 5-cycle appears in the above fixed configuration.x} 
%For each such layer graph of the component, we find the condition or set of conditions that prevents the component from being drawn as a layer graph of that configuration and, a fortiori, the embedded 7-cycle $(p_1, q_1, r_1, s, r_2, q_2, p_2)$ in the configuration of Fig.~\ref{configa:label}. 
This new set of conditions acts as a replacement for the condition $|p_1p_2| \ne |q_2r_2|$ since that set will prevent the drawing of the layer graph of the 7-cycle $(p_1, q_1, r_1, s, r_2, q_2, p_2)$ in the corresponding configuration in Fig.~\ref{configa:label}.    

\begin{comment}
\begin{figure}[htb]
\fbox{\parbox[b]{.99\linewidth}{
\vskip 0.5cm
\centerline{\includegraphics[scale=0.7]{configa.pdf}}
\vskip 0.5cm}}
\caption{\protect\label{configa:label} The layer graph representation of the 7-cycle $(p_1, q_1, r_1, s, r_2, q_2, p_2)$ corresponding to the condition $|p_1p_2| \ne |q_2r_2|$.}
\end{figure}
\end{comment}

%\begin{comment}
\begin{figure}[!h] 
\centering 
\includegraphics[scale=0.7]{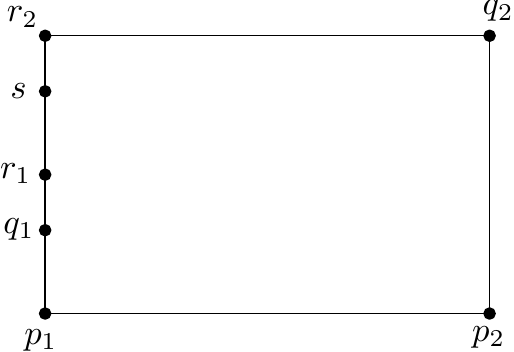} 
\caption{The layer graph representation of the 7-cycle $(p_1, q_1, r_1, s, r_2, q_2, p_2)$ corresponding to the condition $|p_1p_2| \ne |q_2r_2|$.}
\label{configa:label} 
\end{figure} 
%\end{comment}

Since $p_1$, $p_2$ and $p_3$ are made line rigid in the first round they must lie on a line and their positions must be unique (upto translation and reflection) after first round. Since in the present configuration of the 7-cycle (Fig.~\ref{configa:label}) $p_1$ and $s$ are on  the same side of the layer graph the edges $p_3q_3$, $q_3r_3$ and $r_3s$ can have 4 distinct configurations giving rise to 4 distinct layer graph representations (Fig.~\ref{configaall4fig:label}) of the whole component with the layer graph of the 7-cycle being in the configuration of Fig.~\ref{configa:label}.  
%Thus, there will be 2 distinct configurations of the layer graph for the whole jewel as shown in Fig.~\ref{repcon-55jewel-a:label}. 
Thus, in order to be able to draw the layer graph of the 7-cycle $(p_1, q_1, r_1, s, r_2, q_2, p_2)$ in the configuration of Fig.~\ref{configa:label} the layer graph of the whole component must have one of the four distinct configurations as shown in Fig.~\ref{configaall4fig:label}.

%We consider the group where four edges $p_1q_1$, $q_1r_1$, $r_1s$ and $sr_2$ are on one side of the layer graph of the 7-cycle. 
First, we consider the configuration where $p_3q_3$ and $r_3s$ are horizontal, and $q_3r_3$ is vertical (Fig.~\ref{configaall4fig:label}a). The condition $|p_1p_2| \ne |q_2r_2|$  prevents the 7-cycle from being drawn as a layer graph of present configuration. However, it involves the edge $q_2r_2$ which we need to avoid. In the present configuration of the layer graph of the component $p_1, q_1, r_1, s$ and $r_2$ are on a line which is parallel to $p_2q_2$ and $q_3r_3$. So, we must have $|q_2r_2| = ||p_2p_3| \pm |p_3q_3| \pm |r_3s||$. Using this the condition becomes $|p_1p_2| \ne ||p_2p_3| \pm |p_3q_3| \pm |r_3s||$. Since $||p_1p_2| \pm |p_2p_3|| = |p_1p_3|$ the condition reduces to $|p_1p_3| \ne ||p_3q_3| \pm |r_3s||$. If we ensure this condition then we must have $|p_1p_2| \ne |q_2r_2|$ in the present configuration of the component. Thus, the component in general and the 7-cycle in particular cannot be drawn as a layer graph in the present configurations of the 7-cycle and the component.
%We want to make the component line rigid irrespective of the values of the distances $|q_1r_1|$, $|q_2r_2|$ and $|q_3r_3|$. This will allow us to query the distances in such a way that the rigidity conditions are satisfied. For that we see that we can make the 7 cycle $p_1q_1r_1sr_3q_3p_3$ line rigid by imposing the condition $|p_3q_3| \pm |r_3s| \ne |p_1p_3|$. This will fix $s$. Since $p_2$ is also fixed we can consider $p_2q_2r_2s$ as a 4-cycle. In this cycle clearly $|p_2q_2| \ne |r_2s|$. Thus, the points $q_2$ and $r_2$ will be unambiguous.

\begin{comment}
\begin{figure*}[htb]
\fbox{\parbox[b]{.99\linewidth}{
\vskip 0.5cm
\centerline{\includegraphics[scale=0.7]{configa-all4.pdf}}
\vskip 0.5cm}}
\caption{\protect\label{configaall4fig:label} The layer graph of the 7-cycle $(p_1, q_1, r_1, s, r_2, q_2, p_2)$ has 4 edges $p_1q_1$, $q_1r_1$, $r_1s$ and $sr_2$ on one side.}
\end{figure*}
\end{comment}

\begin{figure}[!h] 
\centering 
\includegraphics[scale=0.7]{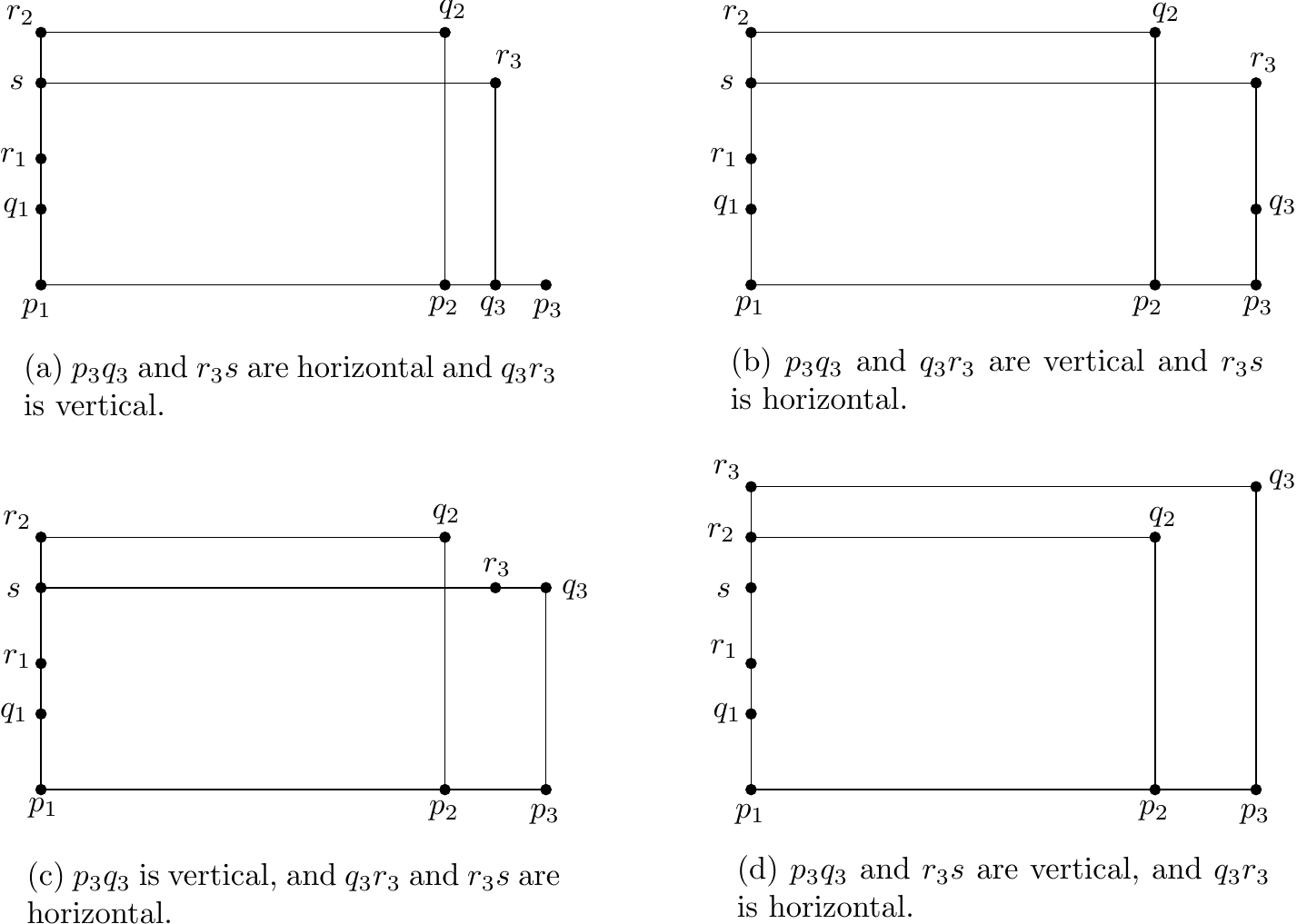} 
\caption{Layer graphs of the basic component when the layer graph of the 7-cycle $(p_1, q_1, r_1, s, r_2, q_2, p_2)$ has 4 edges $p_1q_1$, $q_1r_1$, $r_1s$ and $sr_2$ on one side.}
\label{configaall4fig:label} 
\end{figure} 

%When $p_3q_3$ and $q_3r_3$ are vertical, and $r_3s$ is horizontal the 7-cycle $(p_1, q_1, r_1, s, r_3, q_3, p_3)$ can be made line rigid by ensuring $|p_1p_3| \ne |r_3s|$ (Fig.~\ref{configaall4fig:label}b). Then $s$ will be unique and consequently as above $q_2$ and $r_2$ will be unique. 
Now we consider the case when $p_3q_3$ and $q_3r_3$ are vertical, and $r_3s$ is horizontal (Fig.~\ref{configaall4fig:label}b). In the present configuration of the layer graph of the component $p_1, q_1, r_1, s$ and $r_2$ are on a line, and $p_3q_3$ and $q_3r_3$ are on a line. Those lines are parallel and they are parallel to $p_2q_2$. So, we must have $|q_2r_2| = ||p_2p_3| \pm |r_3s||$. Using this the condition becomes $|p_1p_2| \ne ||p_2p_3| \pm |r_3s||$. We have $||p_1p_2| \pm |p_2p_3|| = |p_1p_3|$. Using this the rigidity condition $|p_1p_2| \ne |q_2r_2|$  becomes $|p_1p_3| \ne |r_3s|$.  %If we ensure this condition then we must have $|p_1p_2| \ne |q_2r_2|$ in the present configuration of the component. Thus, the component in general and the 7-cycle in particular cannot be drawn as a layer graph in the present configurations of the 7-cycle and the component.
%the 7-cycle $(p_1, q_1, r_1, s, r_3, q_3, p_3)$ can be made line rigid by ensuring $|p_1p_3| \ne |r_3s|$ (Fig.~\ref{configaall4fig:label}b). Then $s$ will be unique and consequently as above $q_2$ and $r_2$ will be unique. 

%\begin{figure}[!h] 
%\centering 
%\includegraphics[scale=1.0]{configa-2.pdf} 
%\caption{4 edges on one side of layer graph.}\label{configa2fig:label} 
%\end{figure} 
%Similarly, when $p_3q_3$ is vertical, and $q_3r_3$ and $r_3s$ are horizontal the 7-cycle $p_1q_1r_1sr_3q_3p_3$ can be made line rigid by satisfying the condition $|p_3q_3| \ne |p_2q_2| \pm |r_2s|$ (Fig.~\ref{configaall4fig:label}c). And this will also make $q_2$ and $r_2$ line rigid. Finally, when $p_3q_3$ is vertical, $q_3r_3$ is horizontal and $r_3s$ is vertical the 7-cycle $p_1q_1r_1sr_3q_3p_3$ can be made line rigid by the condition $|p_3q_3| \pm |r_3s| \ne |p_2q_2| \pm |r_2s|$ (Fig.~\ref{configaall4fig:label}d). Then $q_2$ and $r_2$ will also be line rigid.

Next, we consider the case when $p_3q_3$ is vertical, and $q_3r_3$ and $r_3s$ are horizontal (Fig.~\ref{configaall4fig:label}c). The condition $||p_1q_1| \pm |q_1r_1| \pm |r_1s| \pm |sr_2|| \ne |p_2q_2|$  prevents the 7-cycle from being drawn as a layer graph of present configuration. However, it involves the edge $q_1r_1$ which we need to avoid. In the present configuration of the layer graph of the component $p_1$, $p_2$ and $p_3$ are on a line, and $q_3$, $r_3$ and $s$ are on a line. The lines are parallel. So, we must have $||p_1q_1| \pm |q_1r_1| \pm |r_1s|| = |p_3q_3|$. Using this the condition becomes $||p_3q_3| \pm |sr_2|| \ne |p_2q_2|$. 
% the 7-cycle $p_2q_2r_2sr_3q_3p_3$ can be made line rigid by satisfying the condition $|p_3q_3| \ne |p_2q_2| \pm |r_2s|$ (Fig.~\ref{configaall4fig:label}c). Since $p_1$ is also fixed we can consider $p_1q_1r_1s$ as a 4-cycle. To make this cycle line rigid we need to satisfy the condition $|p_1q_1| \ne |r_1s|$. 

Finally, we consider the case when $p_3q_3$ is vertical, $q_3r_3$ is horizontal and $r_3s$ is vertical (Fig.~\ref{configaall4fig:label}d). 
%The condition $||p_1q_1| \pm |q_1r_1| \pm |r_1s| \pm |sr_2|| \ne |p_2q_2|$  prevents the 7-cycle from being drawn as a layer graph of present configuration. However, it involves the edge $q_1r_1$ which we need to avoid. 
 In the present configuration of the layer graph of the component $p_1$, $p_2$ and $p_3$ are on a line. The line is parallel to $q_3r_3$. So, we must have $||p_1q_1| \pm |q_1r_1| \pm |r_1s| \pm |sr_3|| = |p_3q_3|$. Using this the rigidity condition $||p_1q_1| \pm |q_1r_1| \pm |r_1s| \pm |sr_2|| \ne |p_2q_2|$ becomes $||p_3q_3| \pm |sr_3| \pm |sr_2|| \ne |p_2q_2|$.
% the 7-cycle $p_2q_2r_2sr_3q_3p_3$ can be made line rigid by the condition $|p_3q_3| \pm |r_3s| \ne |p_2q_2| \pm |r_2s|$ (Fig.~\ref{configaall4fig:label}d). As we need to satisfy the condition $|p_1q_1| \ne |r_1s|$ to make the 4-cycle $p_1q_1r_1s$ line rigid.

The following lemma justifies the replacement (the proof is omitted). 
  
\begin{lemma} \label{repconp1p2neq2r2lemma:label}  
The 7-cycle $(p_1, q_1, r_1, s, r_2, q_2, p_2)$ of the 3-path basic component of Fig.~\ref{3pathfig:label} cannot be drawn as the layer graph 
of Fig.~\ref{configa:label} if the edges of the component satisfy the following conditions:  

%\begin{multline}\label{eqn:repconlemma}
%\{|p_1p_3| \ne ||p_3q_3| \pm |r_3s||, |p_1p_3| \ne |r_3s||, ||p_3q_3| \pm |sr_2||\ne |p_2q_2|,\\
%||p_3q_3| \pm |sr_2| \pm |sr_3||\ne |p_2q_2|\}  
%\end{multline}

%\begin{equation}{\label{eqn:repconlemma}}  
%\{|p_1p_3| \ne ||p_3q_3| \pm |r_3s||, |p_1p_3| \ne |r_3s||, ||p_3q_3| \pm |sr_2||\ne |p_2q_2|, ||p_3q_3| \pm |sr_2| \pm |sr_3||\ne |p_2q_2|\}  
%\end{equation}  

\begin{center}
$|p_1p_3| \ne ||p_3q_3| \pm |r_3s||, |p_1p_3| \ne |r_3s||, ||p_3q_3| \pm |sr_2||\ne |p_2q_2|,$

$||p_3q_3| \pm |sr_2| \pm |sr_3||\ne |p_2q_2|$
\end{center}
\end{lemma}

Similarly, we can replace the other conditions for rigidity that involve the edges $q_1r_1$ and $q_2r_2$.  Collecting all the conditions we have the following lemma:

\begin{lemma}\label{3pathlemma:label}
The structure consisting of 3 paths of degree 2 nodes $p_1q_1r_1s$, $p_2q_2r_2s$ and $p_3q_3r_3s$ of length 2 attached to the common node $s$ of degree 3 and having the nodes $p_1$, $p_2$ and $p_3$ fixed in the first round is line rigid if its edges satisfy the following conditions:

\begin{enumerate}
%3 conditions on p1p2
\item $|p_1p_2| \notin $ 
\{$|r_1s|$,
$|r_2s|$,
$||r_1s| \pm |r_2s||$\},

%3 conditions on p2p3
\item $|p_2p_3| \notin $ 
\{$|r_2s|$,
$|r_3s|$,
$||r_2s| \pm |r_3s||$\},

%3 conditions on p3p1
\item $|p_3p_1| \notin $ 
\{$|r_3s|$,
$|r_1s|$,
$||r_3s| \pm |r_1s||$\},

%9 conditions on p1q1
\item $|p_1q_1| \notin$ 
\{$|r_1s|$,
$|r_2s|$,
$||r_1s| \pm |r_2s||$,
$||p_1p_2| \pm |r_1s||$,
$||p_1p_2| \pm |r_2s||$,
$||p_1p_3| \pm |r_1s||$,
$||p_1p_3| \pm |r_3s||$,
$||p_1p_2| \pm |r_1s| \pm |r_2s||$,
$||p_1p_3| \pm |r_1s| \pm |r_3s||$\},

%16 conditions on p2q2
%\noindent
\item $|p_2q_2| \notin$
\{$ |r_1s|$,
$|r_2s|$,
$|p_1q_1|$,
$||r_1s| \pm |r_2s||$,
$||p_1p_2| \pm |r_1s||$,
$||p_1p_2| \pm |r_2s||$,
$||p_2p_3| \pm |r_2s||$,
$||p_2p_3| \pm |r_3s||$,
$||p_1q_1| \pm |r_1s||$,
$||p_1q_1| \pm |r_2s||$,
$||p_1p_2| \pm |r_1s| \pm |r_2s||$,
$||p_2p_3| \pm |r_2s| \pm |r_3s||$,
$||p_1q_1| \pm |r_1s| \pm |r_2s||$,
$||p_1q_1| \pm |p_1p_2| \pm |r_1s||$,
$||p_1q_1| \pm |p_1p_2| \pm |r_2s||$,
$||p_1q_1| \pm |p_1p_2| \pm |r_1s| \pm |r_2s||$\},

%21 conditions on p3q3
%\noindent
\item $|p_3q_3| \notin$
\{$ |r_1s|$,
$|r_2s|$,
$|r_3s|$,
$|p_1q_1|$,
$|p_2q_2|$,
$||r_2s| \pm |r_3s||$,
$||r_3s| \pm |r_1s||$,
$||p_1p_3| \pm |r_3s||$,
$||p_2p_3| \pm |r_3s||$,
$||p_1q_1| \pm |r_1s||$,
$||p_1q_1| \pm |r_3s||$,
$||p_2q_2| \pm |r_2s||$,
$||p_2q_2| \pm |r_3s||$,
$||p_1p_3| \pm |r_1s| \pm |r_3s||$,
$||p_2p_3| \pm |r_2s| \pm |r_3s||$,
$||p_1q_1| \pm |r_1s| \pm |r_3s||$,
$||p_2q_2| \pm |r_2s| \pm |r_3s||$,
$||p_1q_1| \pm |p_1p_3| \pm |r_3s||$,
$||p_2q_2| \pm |p_2p_3| \pm |r_3s||$,
$||p_1q_1| \pm |p_1p_3| \pm |r_1s| \pm |r_2s||$,
$||p_2q_2| \pm |p_2p_3| \pm |r_2s| \pm |r_3s||$\}.
\end{enumerate}
\end{lemma}

As mentioned before, we make triplet of points $(p_1, p_2, p_3)$ of each 3-path component line rigid in the first round. Let $S$ be the set of points for such triplets. We make the points in $S$ line rigid in the first round. We make the remaining 7 points of each 3-path component line rigid in the second round. To select triplet of points in $S$ as $(p_1, p_2, p_3)$ of a component, let us select any point of $S$ as $p_1$. Then let us find another point of $S$, we denote it as $p_2$, satisfying the conditions on the length $|p_1p_2|$ mentioned in serial number 1 of Lemma~\ref{3pathlemma:label}. By Observation 1, at most 8 edges will not satisfy the conditions on $|p_1p_2|$. We need at least 8 extra points, i.e., we need to have a total of at least 9 more points, other than $p_1$, in $S$ as candidate for $p_2$. 

After $p_2$ is selected, let us find another point of $S$, we denote it as $p_3$, from the remaining pints of $S$ such that the conditions on $|p_2p_3|$ in serial numbers 2 of Lemma~\ref{3pathlemma:label} are satisfied. By Observation 1, at most 8 edges will not satisfy the conditions on $|p_2p_3|$. This warrants the set $S$ to have at least 8 extra points other than $p_1$, $p_2$ and $p_3$. The point $p_3$ selected this way by satisfying the conditions on $p_2p_3$ must also have to satisfy the conditions on $p_3p_1$ mentioned in serial number 3 of Lemma~\ref{3pathlemma:label}. By Observation 1, at most 8 edges will not satisfy the conditions on $|p_3p_1|$. This warrants the set $S$ to have at least 8 more extra points, i.e., 16 extra points, other than $p_1$, $p_2$ and $p_3$. 

But if $S$ has only 19 points for the selection of $p_i$s it may happen that all the basic components are attached to the same triplets. This hinders our goal of obtaining a better value for $\alpha$ than previously known. We need to attach the basic components evenly to all the points of $S$ so that the same number of edges can be attached to each of them in the first round and all of those edges, except for a constant number, are used to attach the basic components. In other words, we need to attach the 3-path components to the points in $S$ in such a way that the numbers of components attached to any two points differ by at most a constant number. 

Now we describe our algorithm to select triplets of points in $S$ to attach components. To attach a basic component we always select a point in $S$ with the lowest valence as the first point (say $p_1$). Of the remaining points of $S$, at most 8 points may not be acceptable for the second point (say $p_2$), because of the conditions on $p_1p_2$. From among the rest $|S|-1$ points that satisfy the conditions on $p_1p_2$ we select the one that has the lowest valence, as $p_2$. Of the rest $|S|-2$ points of $S$, at most 16 may not be acceptable for the last point, say $p_3$, because of the conditions on $p_2p_3$ and $p_3p_1$. From among the rest points that satisfy the conditions on $p_2p_3$ and $p_3p_1$ we choose the one that has the lowest valence, as $p_3$. This method will be follwed to attach each basic component to the points in $S$. While, we shall attach the basic components sequentially.

To specify the number of basic components attached to a point in $S$ we shall use the term valence. We denote the set of points with valence $d$ as $S_d$. The following lemma tells us how big $S$ must be (the proof is omitted):

%Then we can attach the first 3 basic structures to $p_i$'s in such a way that each of a group of 9 points are not on any basic structure, each of another group of 9 points are on exactly one basic structure and each of the rest 9 points are on exactly 2 basic structure. For the successive basic structures we can find 1 point from each group for being used as $p_1$, $p_2$ and $p_3$ satisfying the above conditions on $|p_1p_2|$, $|p_2p_3|$ and $|p_3p_1|$. In addition to the 162 edges needed at each of $p_i$'s to satisfy the conditions on $|p_1q_1|$, $|p_2q_2|$ and $|p_3q_3|$ we need 2 more edges to accomodate this difference of the number of basic structures that can be attaced to them. Thus, we need a total of 164 extra edges at each of the fixed points $p_i, i = 1, ..., 35$.\\

\begin{lemma} \label{fixedpointlemma:label}
%35 fixed points are sufficient to attach the basic components evenly to them. 
A set $S$ of 35 points is sufficient to ensure that the valences of any two points in $S$ differ by at most 2. 
\end{lemma}

%In addition to the 162 edges needed at each of $p_i$'s to satisfy the conditions on $|p_1q_1|$, $|p_2q_2|$ and $|p_3q_3|$ we need 2 more edges to accommodate this difference of the number of basic components that can be attached to them. Thus, we need a total of 164 extra edges at each of the fixed points $p_i, i = 1, ..., 35$. In the next section we show how to construct a composite $ppg$ made up of 3-path components such that all the rigidity conditions listed in Lemma~\ref{3pathlemma:label} are satisfied for each one of these.   

We make the above set $S$ of 35 points line rigid in the first round by using jewel of Damaschke~\cite{DBLP:journals/dam/Damaschke03} as the \textit{ppg}. We create 6 jewels hanging from a common strut that is incident on 2 points of $S$. This will make 32 points line rigid. For this we need to query the lengths of 49 edge. We make the remaining 3 points line rigid by using triangle as the \textit{ppg}. For each of these 3 points we query its distance from each of the pair of points that are incident on the strut. There will be 6 more queries for edge lengths. Thus, we shall query a total of 55 edges in the first round to make the 35 points of $S$ line rigid in that round. %Since each 4:4 jewel is line rigid so is this configuration. We will call the points in $S$ fixed  since we can fix their placement on a line by querying the edges of this $ppg$. \\    

The conditions on $p_1q_1$, $p_2q_2$ and $p_3q_3$ in serial numbers respectively 3, 4 and 5 of Lemma~\ref{3pathlemma:label} will not be satisfied by at most 40, 90 and 122 edges respectively (by Observation 1). In addition to the 122 extra edges needed at each of $p_i$'s to satisfy the conditions on $|p_1q_1|$, $|p_2q_2|$ and $|p_3q_3|$ we need 2 more extra edges incident on each of $p_i$ to accommodate the difference of 2 between the number of basic components that can be attached to the $p_i$'s. Thus, we need a total of 124 extra edges incident on each of the points $p_i, i = 1, ..., 35$ of $S$. We shall attach $3b$, $3b+1$ or $3b+2$ (where $b$ is a positive integer) number of 3-path components to each point in $S$. This requires us to have $3b+124$ edges incident on each of $p_i$'s in $S$.  In the worst case there will be at most 18 points in $S$ with valence $3b$, no points in $S$ with valence $b+1$ and the remaining points with valence $3b+2$. Thus, we shall be able to construct a total of at least $3b+11$ number of 3-path components from the edges provided for $p_iq_i$ at all the $p_i$'s in $S$. Now we describe the algorithm to construct a composite $ppg$ made up of 3-path components such that all the rigidity conditions listed in Lemma~\ref{3pathlemma:label} are satisfied for each of them.

{\bf Algorithm 1.}
Let the total number of points be $n = 245b + 4,419$, where $b$ is a positive integer. We attach at least $3b$ and at most $3b+2$ numbers of 3-path components (Fig.~\ref{3pathfig:label}) to each of 35 rigid points in $S$ subject to the condition that the total number of such components being $35b + 11$. 
  
In the first round, we make distance queries represented by the edges of the graph in Fig.~\ref{querygraph:label}. All the nodes $p_i$ ($i=1, ..., 35$) in the subgraph enclosed by the rectangle are elements of $S$ and are made line rigid in the first round by using the jewel of~\cite{DBLP:journals/dam/Damaschke03} as the \textit{ppg}. There are 6 jewels attached to a common strut in the subgraph. Residual 3 points are made line rigid by using triangle as the \textit{ppg}. They are attached to the common strut. There are a total of 55 edges in the subgraph. Each of the vertices $p_i, p_j$, or $p_k$  ($i, j, k = 1, ..., 35$) of $S$ has $b+124$ leaves to attach $3b$, $3b+1$ or $3b+2$ 3-path components (Fig.~\ref{3pathfig:label}). Since there will be $35b+11$ 3-path components we make $35b+11$ groups of 4 nodes $(r_{il}, r_{jl}, r_{kl}, s_l)$, $(l = 1, ..., 35b+11)$. We query the distances $|r_{il}s_l|$, $|r_{jl}s_l|$ and $|r_{kj}s_l|$, $(l = 1, ..., 35b+11)$  
in the first round. We will make a total of $210b+4,428$ pairwise distance queries in the first round for the placement of $n = 245b + 4,419$ points. 

\begin{figure}[!h] 
\centering 
\includegraphics[scale=0.7]{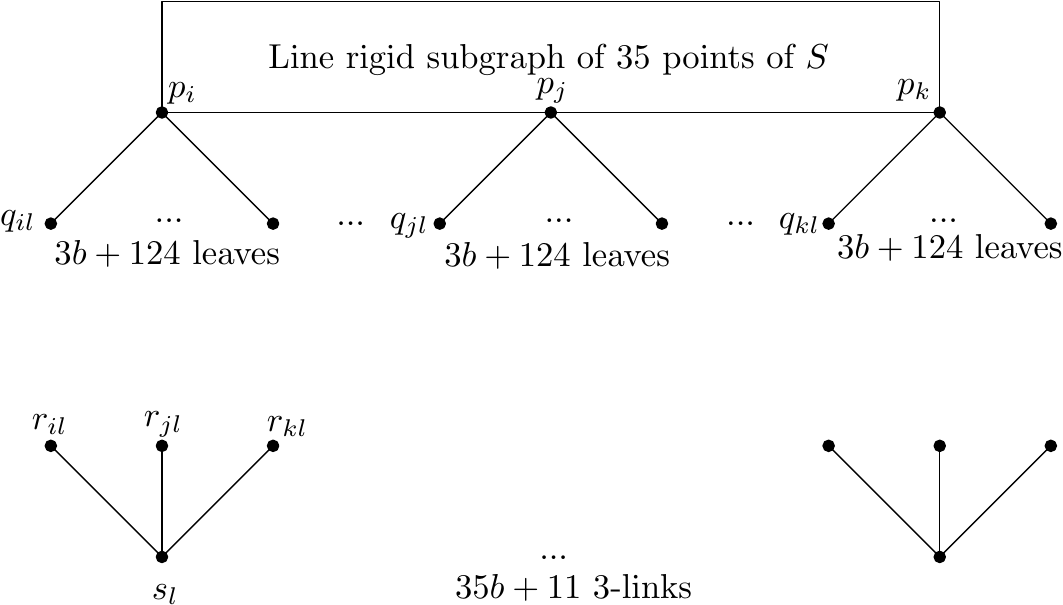} 
\caption{Queries in the first round.}
\label{querygraph:label} 
\end{figure}

In the second round, for each 3-link $(r_{il}, r_{jl}, r_{kl}, s_l), l = 1, ..., 35b + 11,$ we construct a 3-path component (Fig.~\ref{3pathfig:label}),  satisfying all its rigidity conditions as in Lemma~\ref{3pathlemma:label}. For each such 3-link we  select a point $p_i$, from the subgraph of 35 points of $S$ that has the lowest valency of 3-path component of Fig.~\ref{3pathfig:label}. Since all the 35 points $p_i, i = 1, ..., 35,$ are rigid in the first round, for any pair of such fixed points $(p_i, p_j) (i, j = 1, ...35; i \ne j)$ we can find the distance $|p_ip_j|$. So, for each pair of points $(p_i, p_j) (i, j = 1, ..., 35; i \ne j)$, we shall use  $(p_i, p_j)$ as an edge in the construction of the 3-path component of Fig.~\ref{3pathfig:label}. 

Now from the subgraph of 35 points of $S$ we select another point $p_j (j \ne i)$ such that the length $|p_ip_j|$ satisfies all the 4 conditions of rigidity on it as stated in serial number 1 of 
Lemma~\ref{3pathlemma:label} and that it has the lowest valency of 3-path component of Fig.~\ref{3pathfig:label} among all such qualifying points. We note that we can always find such point $p_j$, because there will be at most 8 edges $(p_ip_j)$ whose lengths do not satisfy the rigidity conditions on it (Lemma~\ref{3pathlemma:label}) whereas we have 34 more points for choosing the point $p_j$. Similarly, from the subgraph of 35 points of $S$ we select another point $p_k (k \ne i, k \ne j)$ such that the length $|p_jp_k|$ satisfies all the 4 conditions of rigidity on it as stated in serial number 2 of Lemma~\ref{3pathlemma:label} and the length $|p_kp_i|$ satisfies all the 4 conditions of rigidity on it as stated in serial number 3 of Lemma~\ref{3pathlemma:label}, and that it has the lowest valency of 3-path component of Fig.~\ref{3pathfig:label} among all such qualifying points. We note that we can always find such point $p_k$, because there will be at most 16 nodes $p_k$ such that the lengths of the edges $p_jp_k$ and $p_kp_i$ do not satisfy the rigidity conditions on them (Lemma~\ref{3pathlemma:label}) whereas we have 33 more points for choosing the point $p_k$.

Then we find an edge $p_iq_{il}$ rooted at $p_i$ satisfying the 20 conditions of rigidity on it as stated in serial no. 4 of Lemma~\ref{3pathlemma:label}, then we find another edge $p_jq_{jl}$  rooted at $p_j$ satisfying the 45 conditions on it as stated in serial no. 5 of  Lemma~\ref{3pathlemma:label} and finally, we find another edge $p_kq_{kl}$ rooted at $p_k$ satisfying the 61 conditions on it as stated in serial no. 6 of Lemma~\ref{3pathlemma:label}.

Then for each $l, (l = 1, ..., 35b + 11)$, we query the distances $|q_{il}r_{il}|$, $|q_{jl}r_{jl}|$  and $|q_{kl}r_{kl}|$ to form a 3-path component $p_ip_jp_kq_{il}q_{jl}q_{kl}r_{il}r_{jl}r_{kl}s_l$. Its edges will satisfy all the rigidity conditions of Lemma~\ref{3pathlemma:label}. Thus, all the $35b + 11$ 3-links will be consumed to construct $35b + 11$ 3-path components. For this $105b + 33$ edges will be queried in the second round.

There will be unused leaves $q_{il}$/$q_{jl}$/$q_{kl}$) numbering 4,307 in total for the 35 points of $S$. We use a 4-cycle \textit{ppg}~\cite{DBLP:journals/dam/Damaschke03} to fix 4,306 of them and a triangle \textit{ppg} to fix the rest 1 point in the second round. As before, for each pair of points $(p_i, p_j) (i, j = 1, ..., 35; i \ne j)$, we shall use  $(p_i, p_j)$ as an edge in the construction of the 4-cycle. For each unused point $q_{il}$ rooted at $p_i$ we find another point $q_{jl}$ rooted at $p_j$ such that $|p_ip_{il}| \ne |p_jp_{jl}|$. Then the 4-cycle $p_iq_{il}q_{jl}p_j$ will be line rigid (Observation 2). Then we query the distance $|q_{il}q_{jl}|$ in the second round to complete the 4-cycle.  Note that we can always find a point like $q_{jl}$. For, after repeated selection of such matching pairs of edges there may remain at most 2 edges $p_iq_{il}$ rooted at $p_i$ of length equal to that of the same number of edges rooted at $p_j$ (Observation 1). In such a situation we switch the matching to match such edges rooted at $p_i$ with edges other than those same length edge/s rooted at $p_j$ - this is always possible because there are at most 2 edges rooted at $p_j$ that have the same length (Observation 1). To make the remaining 1 leave node line rigid we query in the second round its distance from any point of $S$ other than its parent node. 

For 4,307 unused points (after the construction of the 3-path components) 2,153 4-cycles and 1 triangle will be constructed.  2,153 edges will be queried to complete the 4-cycles and 1 edge will be queried to construct the triangle. The total number of queries in the second round will be $(105b+33)+2,153+1$, i.e., $105b+2,187$. \qed

\begin{theorem}\label{3pathrigiditytheorem:label}  
The \textit{ppg} constructed by Algorithm 1 is line rigid.  
\end{theorem}  
  
\begin{proof}  
Omitted.  
\end{proof}

The number of queries in the first and second rounds are $210b + 4,428$ and $105b + 2,187$ respectively. Thus, in 2 rounds a  
total of $315b + 6,615$ pairwise distances are to be queried for the placement of $245b + 4,419$ points. %It is interesting to note that our algorithm would need at least ??? points to work, which makes it reasonably practical. When we have fewer points we can use Algorithm ??? instead.  \\  
Now, $315b + 6,615 = (315/245)*(245b + 4419) - (9/7)*4419 + 6615 = 9n/7 + (46305 - 39771)/7 = 9n/7 + 6534/7$. Thus, we have the  
following theorem:

\begin{theorem}\label{3pathUBtheorem:label}  
$9n/7 + 6534/7 $ queries are sufficient to place $n$ distinct points on a line in two rounds.  
\end{theorem}

\section{Lower Bound for Two Rounds} 
 
The argument here closely follows the adversarial argument given in the lower bound proof of 
\cite{DBLP:conf/wabi/ChinLSY07}. %Our main contribution is the strengthening of a key lemma below that improves the lower bound. 
Let the set of edges queried in the first and second round be $E_1$ and $E_2$ respectively; 
$G_1 = (V,E_1)$ is the query graph for the first round, while $G_2 = (V,E_1\cup E_2)$ is the final query graph after the second round. The length of a maximal path of degree 2 nodes in a graph is the number of degree 2 nodes in the path. We shall call nodes of degree at least 3 as heavy nodes.
  
In the \emph{first round}, the adversary returns edge-lengths according to the following strategy, 
with the intention of keeping the linear layout of the \textit{ppg} ambiguous: 
  
\begin{description} 

\item{$S_1$:} The adversary fixes the layout of all nodes of degree 3 or more and returns the lengths of the edges incident on these nodes. 

%\item{$S_2$:} Let us consider all nodes of degree 2, one of whose incident edges is also incident on a degree 1 node. For each of such nodes, set the length of one of these incident edges to be the same, say $c$, over all these degree 2 nodes. 

\item{$S_2$:} For all degree 2 nodes, if one of the incident edges is also incident on a degree 1 node, the adversary sets the length of one of the incident edges to be the same, say $c$, over all these degree 2 nodes. 

\item{$S_3$:} For maximal paths formed by 2 or more degree 2 nodes, say $p_1, p_2, ..., p_k (k \geq 2)$, 
let $p_0$ and $p_{k+1}$ be non-degree 2 nodes adjacent to $p_1$ and $p_k$ respectively. The adversary sets $|p_{i-1}p_i| = |p_{i+1}p_{i+2}|$ for $i = 1$ (mod 3). 
%Now we consider the residual case when $p_{k+1}$ is not a heavy node.
In addition, if both $p_0$ and $p_{k+1}$ are of degree 3 or more the adversary sets $|p_{i}p_{i+1}| = |p_{i-1}p_{i+2}|$ for $i = 1$ (mod 3), and if at least one of them, say $p_{k+1}$, is of degree one the adversary sets the lengths of alternate edges equal.
%$|p_kp_{k+1}| = |p_{k-2}p_{k-1}|$.
%\item If a node of degree 3 has 2 maximal paths of degree 2 or 1 nodes the other ends of which are not attached to any heavy node (i.e., node of degree at least 3) and if the node is incident on only one maximal path of degree 2 node of length 1 of which the other end is incident on a node of degree at least 3, set the length of the edge of this path that is incident on this node as $c$.

\item{$S_4$:} If a node, say $p_0$, of degree 3 has 2 maximal paths of degree 2 or 1 nodes the other ends of which are not attached to any heavy node, and if the node $p_0$ is incident on only one maximal path of degree 2 node of length 1 of which the other end is incident on a heavy node, then set the length of one of the edges of this third path as $c$.

%\item If a node, say $p_0$, of degree 3 has 2 maximal paths of degree 2 or 1 nodes the other ends of which are not attached to any heavy node (i.e., node of degree at least 3), and if the node $p_0$ is incident on only one maximal path of degree 2 node of length 1 of which the other end is incident on a heavy node, then set the length of one of the edges of this third path as $c$.

\end{description} 

%The adversary sets the above layout in such a way that if, for any $i$ with $i = 1$ (mod 3) and $i < k$, no edge is attached to either $p_i$ or $p_{i+1}$ in the second round, they can be made ambiguous by setting $|p_{i}p_{i+1}| = |p_{i-1}p_{i+2}|$ in the first or second round if in the first round $p_{k+1}$ is of degree at least 3 or of degree 1 respectively. And if $p_{k+1}$ is of degree 1 and no edge is attached to either $p_{k-1}$ or $p_k$ in the second round the positions of $p_{k-1}$ and $p_k$ can be made ambiguous by setting $|p_{k-1}p_k| = |p_{k-2}p_{k+1}|$ in that round.

%For a maximal path of degree 2 nodes in $G_2$ Step 3 of the adversary imposes limits on the maximum number of edges from $E_1$ if the path consists of edges from $E_1$ only, or the maximum number of consecutive edges from $E_1$ if it contains at least one edge from $E_2$. 
For a maximal path of degree 2 nodes in $G_2$, as a consequence of $S_3$ there are limits on the maximum number of edges from $E_1$ if the path consists of edges from $E_1$ only (Fig.~\ref{deg2maxpath-bothEndHeavy:label} shows a degree 2 maximal path $p_1p_2p_3p_4p_5p_6$ in $G_1$ with both the end nodes $p_0$ and $p_7$ being heavy), and on the maximum number of consecutive edges from $E_1$ if it contains at least one edge from $E_2$ (Fig.~\ref{deg2maxpath-notBothEndHeavy:label} shows some degree 2 maximal paths in $G_1$ with none of the end nodes being heavy). If both of $p_0$ and $p_{k+1}$ are of degree at least three in the first round the adversary sets the above layout in such a way that if, for any $i$ with $i = 1$ (mod 3) and $i < k$, no edge is attached to either $p_i$ or $p_{i+1}$ in the second round their positions will be ambiguous. Thus, for this case the length of a maximal path of degree 2 nodes in $G_2$ containing only the edges in $E_1$ can be at most 3.
If at least one of $p_0$ and $p_{k+1}$, say $p_{k+1}$, is of degree one in the first round the adversary sets the above layout in such a way that if, for any $i$ with $i = 1$ (mod 2) and $i < k$, no edge is attached to either $p_i$ or $p_{i+1}$ in the second round, they can be made ambiguous by setting $|p_{i}p_{i+1}| = |p_{i-1}p_{i+2}|$ in the second round. 
%new on Sep 28
Thus, for this case the length of a maximal path of degree 2 nodes in $G_2$ containing only the edges in $E_1$ can be at most 2.
%And deleted on Sep 28 
If $p_{k+1}$ is of degree 1 and no edge is attached to either $p_{k-1}$ or $p_k$ in the second round the positions of $p_{k-1}$ and $p_k$ can be made ambiguous by setting $|p_{k-1}p_k| = |p_{k-2}p_{k+1}|$ in that round.
%new on Sep 28
The algorithm must attach an edge in $G_2$ to $p_{k-1}$ or $p_k$. Still then there will be at most 2 free nodes at an end of a path of degree 2 nodes if the end node is of degree 1. The algorithm will fix them in the second round. Thus, in a maximal path of degree 2 nodes in $G_2$ that contains at least one edge from $E_2$ there can be at most 2 consecutive edges from $E_1$.

\begin{figure}[!h] 
\centering 
\includegraphics[scale=.7]{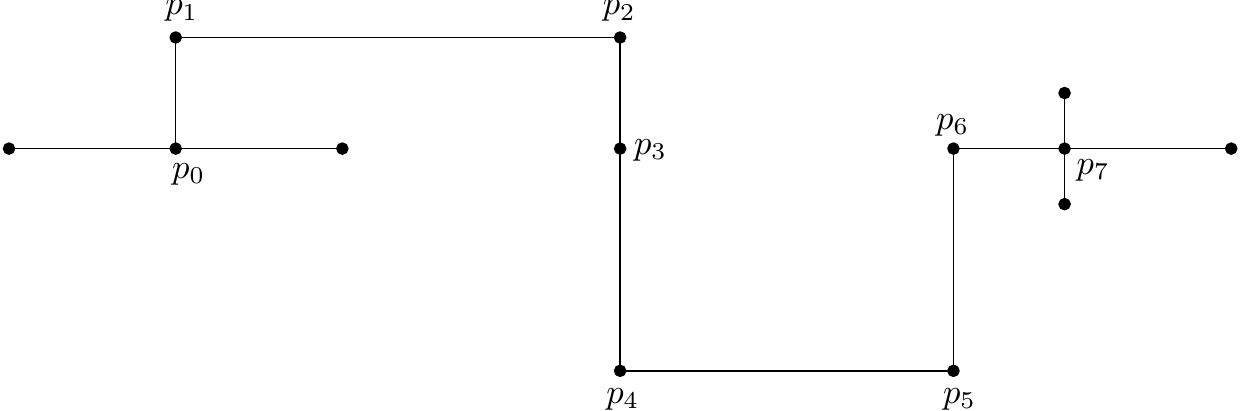} 
\caption{$p_1p_2p_3p_4p_5p_6$ is a maximal path of degree 2 nodes in $G_1$ with both the end nodes being heavy. In the second round, the algorithm has to introduce edges at $p_1$ or $p_2$ to make them unambiguous, and at $p_4$ or $p_5$ to make them unambiguous. This will reduce the length of the degree 2 maximal path in $G_2$.} 
\label{deg2maxpath-bothEndHeavy:label} 
\end{figure} 

\begin{figure}[!h] 
\centering 
\includegraphics[scale=.7]{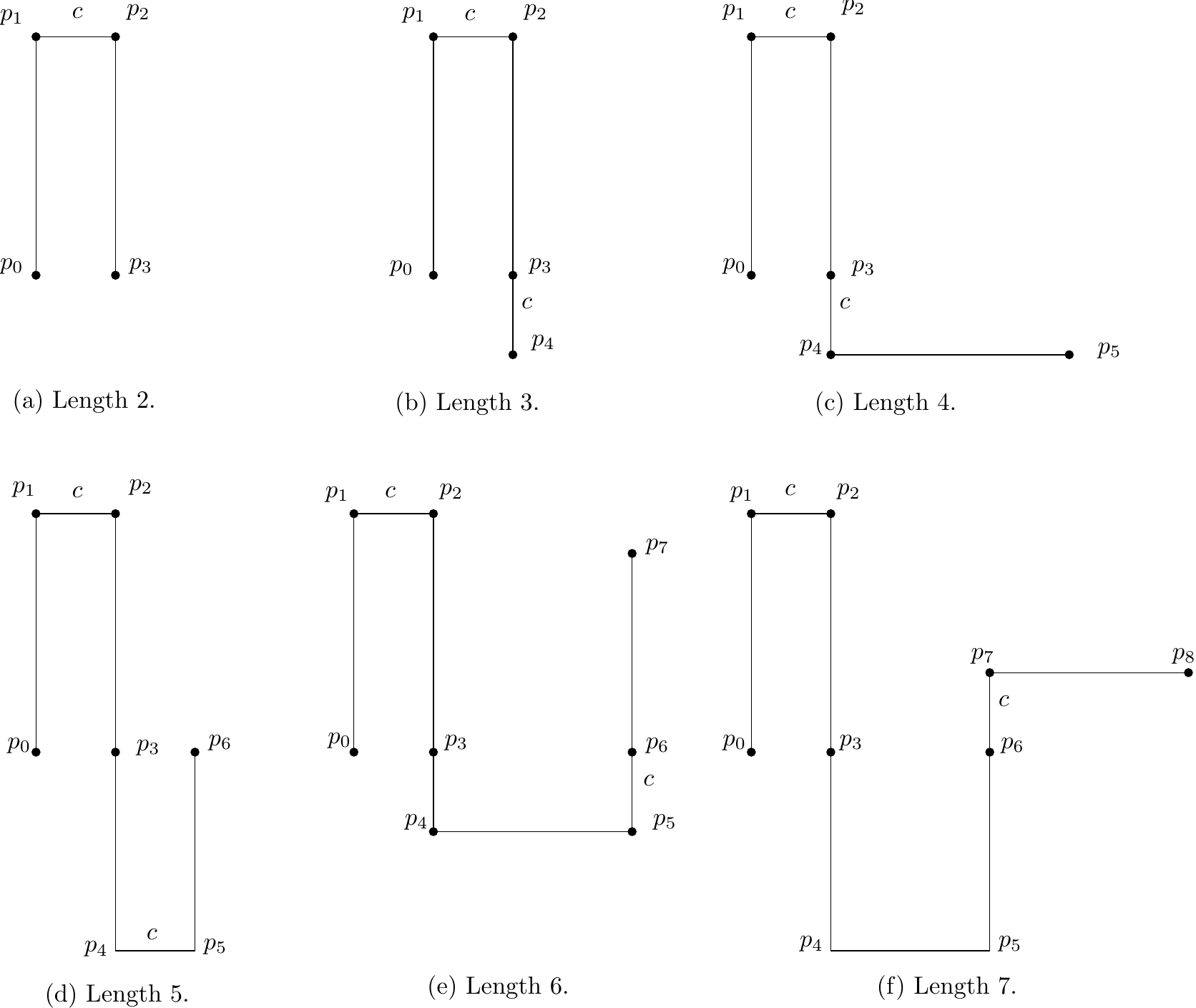} 
\caption{Some maximal paths of degree 2 nodes in $G_1$ with none of the end nodes being heavy. In the second round, the algorithm has to introduce an edge at $p_1$ or $p_2$ in all the figures (a)-(f) to make them unambiguous, and at $p_4$ or $p_5$ in figures (d)-(f) to make them unambiguous.This will reduce the lengths of the degree 2 maximal paths in $G_2$.} 
\label{deg2maxpath-notBothEndHeavy:label} 
\end{figure}

The above results together with $S_2$ and $S_3$ imply that the following property holds for the \textit{ppg}~\cite{am-anewalg}. 
  
\begin{lemma} \label{maxpathlemma:label}
The number of nodes in any maximal path of degree 2 nodes in $G_2$ is at most 3. 
\end{lemma}

%Take, for instance, the combination of edge types $E_2$, $E_1$, $E_2$, $E_1$, $E_1$ as in Fig.~\ref{deg2maxpath-case1label}. Since there are two $E_2$ type edges the lengths of these can be set in such a way that $|p_0p_5| = |p_1p_2| + |p_2p_3|$ and $|p_0p_1| = |p_4p_5|- |p_3p_4|$, making $G_2$ non-rigid. %(Fig.~\ref{deg2maxpath:label}). Similarly, the layout of the next 3 combinations can also be made ambiguous by the adversary. As for the last combination, Step 2 of the adversary's strategy will set the length of one of the edges $p_0p_1$ or $p_1p_2$ equal to that of one of the edges $p_3p_4$ or $p_4p_5$ in round 1. If it sets, for example, $|p_1p_2| = |p_3p_4|$ in round 1, then in round 2 it will set the length of $p_2p_3$ as $|p_2p_3| = |p_4p_5| + |p_5p_0| + |p_0p_1|$  (Fig.~\ref{deg2maxpath-case5label}). Then the cycle $(p_0, p_1, p_2, p_3, p_4, p_5)$ will not be line rigid. \qed

\begin{comment}
\begin{figure}[htb]
\fbox{\parbox[b]{.99\linewidth}{
\vskip 0.5cm
\centerline{\includegraphics[scale=0.8]{deg2maxpath-case1.pdf}}
\vskip 0.5cm}}
\caption{\protect\label{deg2maxpath-case1label} Maximal path of degree 2 in $G_2$ for the combination of edges $E_2, E_1, E_2, E_1, E_1$.}
\end{figure}

\begin{figure}[htb]
\fbox{\parbox[b]{.99\linewidth}{
\vskip 0.5cm
\centerline{\includegraphics[scale=0.8]{deg2maxpath-case5.pdf}}
\vskip 0.5cm}}
\caption{\protect\label{deg2maxpath-case5label} Maximal path of degree 2 in $G_2$ for the combination of edges $E_1, E_1, E_2, E_1, E_1$.}
\end{figure}
\end{comment}

\begin{theorem} 
The minimum density of any line rigid \textit{ppg} for any two round algorithm is at least $\frac{9}{8}$. 
\end{theorem} 

\begin{proof} 
We determine the minimum of the average numbers of edges for all types of nodes. For this 
%the \textit{ppg} is divided into pieces each of which consists of one node and fractions of edges incident on it. To split the edges and allocate their parts between their corresponding adjacent nodes the nodes are categorized 
%as light and heavy. If an edge joins two light nodes or two heavy nodes then the edge is divided equally between the nodes. Otherwise, the light node owns $1/2+g$ and the heavy node owns $1/2-g$ of the edge joining them, where $0 \leq g < 1/2$. 
%According to the construction there are two types of nodes that are analyzed below for their average density: 
the nodes are categorized into two broad types:

\begin{enumerate} 
\item [A.] {\bf Nodes in the maximal paths of length at least 2 formed by degree 2 nodes in the first round where both the end nodes are attached by edges from $E_1$ to nodes of degree at least 3 in the first round:} For edges whose one end is incident on a node of this type and the other end is incident on a node of the other type  the edge is split into 2 equal halves. One half is counted towards the density of the nodes of this path and the other half is counted towards the other type of nodes.

For maximal path of length $k = 2$ the average density is $\frac{1}{2}(2\times\frac{1}{2}+1+\frac{1}{2}) = \frac{5}{4}>\frac{9}{8}$. For $k = 3$ the average is $\frac{1}{3}(2\times\frac{1}{2}+2+\frac{1}{2}) = \frac{7}{6} >\frac{9}{8}$. For $k = 4$ the average is $\frac{1}{4}(2\times\frac{1}{2}+3+\frac{1}{2}) =\frac{9}{8}$. For higher values of $k$ the average is $\frac{1}{k}[2\times\frac{1}{2}+k-1+ \lfloor\frac{k+1}{3}\rfloor\times\frac{1}{2}] > \frac{9}{8}$. Their minimum is $\frac{9}{8}$.
 
%\item [b.] {\bf Nodes in the maximal path formed by degree two nodes:} These are light nodes. By Lemma \ref{maxpathlemma:label} each maximal path of degree two nodes has length $k$, where $k \leq 3$. The total edge weight of such a path is $2(1/2 + g) + (k-1)$. Thus the average density of each node in such a path is $1 + 2g/k$. It is minimum when $k = 3$. Thus, each node has at least $1 + 2g/3$ edges.
\item [B.] {\bf All the remaining Nodes:} To compute the minimum density of this type of nodes we group these nodes and their adjacent edges into neighbourhoods of heavy nodes in $G_2$ of this type and evaluate the average densities of these groups. Their minimum will be the minimum density for this type of nodes. %For grouping around heavy nodes, if a path of degree 2 nodes or an edge is attached to two heavy nodes of this type the path or the edge is divided equally and each half is counted towards the density of one group. To compute the minimum of average densities we shall consider the averages for groups around degree 3 nodes only since averages for groups around higher order nodes will be higher.

There are 2 types of groups around the heavy nodes based on whether the heavy node is connected to a node of type A or a heavy node of type B by a maximal path of degree 2 nodes in $G_2$. Here, the path may have 0 number of degree 2 nodes for which the path will contain only one edge and no degree 2 nodes. If a path of degree 2 nodes or an edge is attached to two heavy nodes of type B the path or the edge is divided equally and each half is counted towards the density of one group. If the two ends of a path are attached to two types of nodes then by the accounting described for type A nodes all the nodes and edges of the path except for the half of the end edge attached to the node of type A are counted towards the density of the group of nodes of type B. We consider the two kinds of nodes of type B separately.

\begin{enumerate}
\item [(a)] {\bf Heavy nodes in $G_2$ that are connected to heavy nodes of type B only, by paths of degree 2 nodes in $G_2$:} 

%First we consider the case when the number of paths $n$ at a heavy node $p_0$ in $G_2$ is more than 3. Let the average number of nodes in each path is $a$. By Lemma~\ref{maxpathlemma:label} we have $a \leq 3$. Average density around $p_0$ is $d = \frac{\frac{1}{2}na + \frac{1}{2}n}{\frac{1}{2}na + 1} = \frac{\frac{1}{2}a + \frac{1}{2}}{\frac{1}{2}a + \frac{1}{n}}$, which is minimum when $n$ has the minimum value 4. For this value of $n$ the average density is $1 + \frac{\frac{1}{4}}{\frac{1}{2}a + \frac{1}{4}}$. This is minimum when $a$ has the maximum value 3. Then the average density is $8/7$ which is greater than $9/8$.

%First we consider the group with all the 3 paths of length 3. For this case the placement will not be unique (Fig.~\ref{333group:label}). There must be an edge at $p_1$ or $p_2$ to make them unambiguous. We count one half of this edge towards the density of the nodes in the path $p_0p_1p_2p_3p_4$ and another half towards the other end. The half of this end is split again into two one quarters of an edge. One quarter of the edge count is added towards the density of the current group. Its average density will be 6.25/5.5=8.33/7.33.

Clearly, for a group of nodes around a heavy node the contribution of average density for the group from an attached path of degree 2 nodes decreases as the length of the path increases. By Lemma~\ref{maxpathlemma:label} the maximum length of a path of degree two nodes is 3. The minimum contribution from a path is $(\frac{3+1}{2})/\frac{3}{2} = \frac{4}{3} > \frac{9}{8}$. Thus, the path will not contribute to reduce the average density of a group around a heavy node to lower than $\frac{9}{8}$. So, we only consider the heavy nodes of this group each of which has the least number of degree two paths attached to the heavy node, i.e., which has exactly 3 paths of degree 2 nodes attached. 

Let the total number of nodes in the 3 paths is $m$. Then average density for the group around the heavy node is $d = \frac{\frac{1}{2}m + \frac{3}{2}}{\frac{1}{2}m + 1} = 1 + \frac{1}{m + 2} \geq \frac{9}{8}$ for $m \leq 6$. Thus, for the groups with paths having total number of degree 2 nodes at most 6 the minumum average density will be $\frac{9}{8}$. It remains to consider the groups with total number of degree 2 nodes 7, 8 and 9, since there can be at most 3 degree 2 nodes in a path by Lemma~\ref{maxpathlemma:label}.

For the group with 2 paths of length 3 and 1 path of length 1 placement will not be unique due to $S_2$ and $S_4$ (Fig.~\ref{331group:label}). For path $p_0p_1p_2$ either $|p_0p_1| = c$ or $|p_1p_2| = c$ by $S_4$. For path $p_0p'_1p'_2p'_3p'_4$, among the edges $p'_1p'_2$ and $p'_2p'_3$ the one in $E_1$ will have length $c$ by $S_2$. Similarly, for the path $p_0p''_1p''_2p''_3p''_4$ either $|p''_1p''_2| = c$ or $|p''_2p''_3| = c$. One more edge must be attached to make the points unique. Then total number of nodes for the paths at the heavy node $p_0$ will be at most 5 and the average density for the group will be at least $\frac{9}{8}$.

\begin{comment}
\begin{figure*}[htb]
\fbox{\parbox[b]{.99\linewidth}{
\vskip 0.5cm
\centerline{\includegraphics[scale=0.8]{331groupa-1.pdf}}
\vskip 0.5cm}}
\caption{\protect\label{331group:label} Heavy node of group B(a) with 2 paths of degree 2 nodes of length 3 and 1 path of degree 2 node of length 1.}
\end{figure*}
\end{comment}

\begin{figure}[!h] 
\centering 
\includegraphics[scale=0.7]{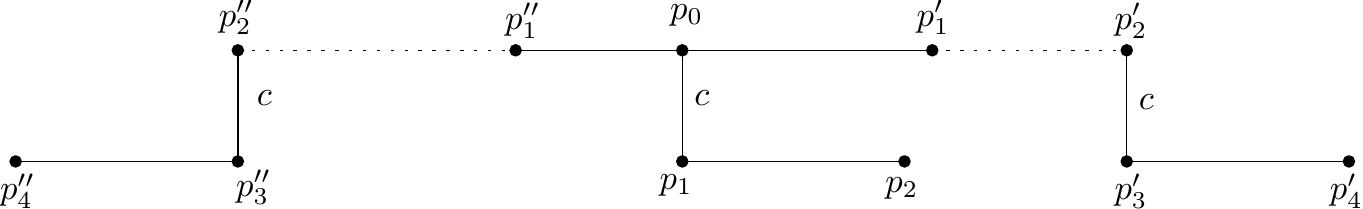}
\caption{Heavy node of group B(a) with 2 paths of degree 2 nodes of length 3 and 1 path of degree 2 node of length 1.} 
\label{331group:label} 
\end{figure}

Now we consider the group with all the 3 paths of length 3. For this case the placement will not be unique (Fig.~\ref{333group:label}). For Fig.~\ref{333group:label}a there must be an edge at $p_1$ or $p_2$ of the path $p_0p_1p_2p_3p_4$ to make $p_1$ and $p_2$  unambiguous. For Fig.~\ref{333group:label}b there must be an edge at $p_2$ or $p_3$ to make $p_2$ and $p_3$ unambiguous. Similarly, there must be an extra edge for each of the other 2 paths $p_0p'_1p'_2p'_3p'_4$ and $p_0p''_1p''_2p''_3p''_4$. Thus, the reduced group consists of 3 degree 2 paths of maximum length 2. There will be at most 6 degree 2 nodes in the degree 2 maximal paths at the heavy node $p_0$ and the average density for the group will be at least $\frac{9}{8}$. Similarly, for the group with 2 paths of length 3 and 1 path of length 2 we can show that there must be edges at the nodes of the paths that will make the total number of nodes in the degree 2 paths around the heavy node at most 6. So, the minimum average density for it will be $\frac{9}{8}$.

\begin{comment}
\begin{figure*}[htb]
\fbox{\parbox[b]{.99\linewidth}{
\vskip 0.5cm
\centerline{\includegraphics[scale=0.8]{333groupa-2-2figs-2.pdf}}
\vskip 0.5cm}}
\caption{\protect\label{333group:label} Heavy node of group B(a) with 3 paths of degree 2 nodes of length 3.}
\end{figure*}
\end{comment}

\begin{figure}[!h] 
\centering 
\includegraphics[scale=0.7]{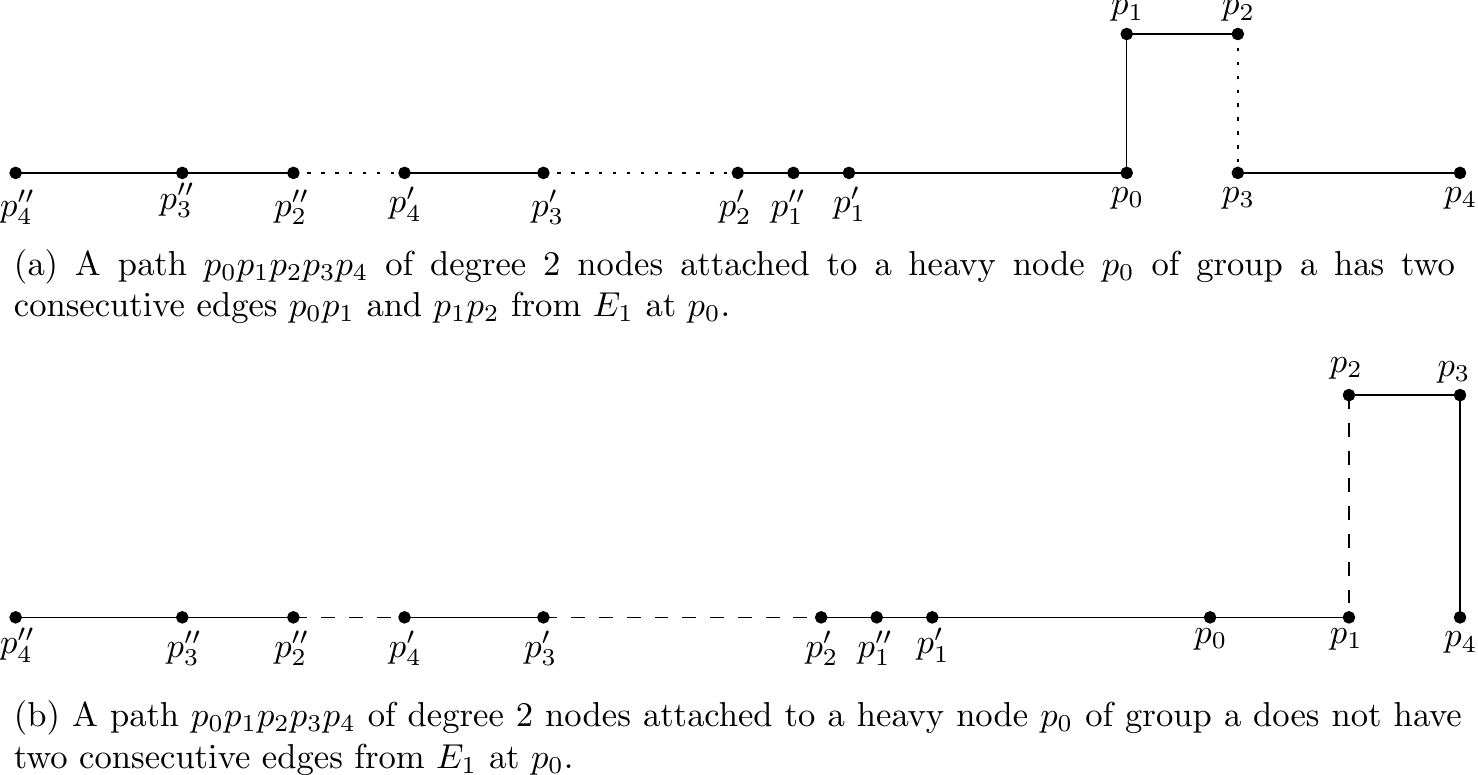} 
\caption{Heavy node of group B(a) with 3 paths of degree 2 nodes of length 3.} 
\label{333group:label} 
\end{figure} 

%Similarly, the average density for the group with 2 paths of lenght 3 and 1 path of length 2 of degree 2 nodes attached to a heavy node will be 7.66/6.66, and for the group with 1 path of length 3 and 2 paths of length 2 the average will be 7/6. For the group with 2 paths of length 3 and 1 path of length 1 placement will not be unique due to adversary strategy 2 and 4 (Fig.~\ref{331group:label}). One more edge must be attached to make them unique. The average density will become 5.25/4.5 = 7/6.

It can be easily seen that for groups with other degree 2 paths attached the minimum average density will be $\frac{9}{8}$. Thus, the minimum of the averages for this type of groups is $\frac{9}{8}$.

\item [(b)] {\bf Heavy nodes in $G_2$ that are connected to at least one node of type A by at least one path of degree 2 nodes in $G_2$:} It is shown above that a maximal path of degree 2 nodes attached to 2 heavy nodes of type B does not contribute to reduce the average density of a group to lower than $\frac{9}{8}$. 
A heavy node of type B in $G_2$ can be connected to a node of type A by a maximal path of degree 2 nodes in $G_2$ in two ways based on whether the edge of the path incident on the type A node is in $E_1$ or $E_2$. If the edge is in $E_1$ the length of the degree 2 path is 0, because type A nodes are found only in the maximal paths of degree 2 nodes in $G_1$ where each end of a path is connected to a heavy node of type B in $G_1$ by an edge from $E_1$. For this case one half of each end edge is counted towards the density of its adjacent node of type B. This path will not contribute to reduce the density of the corresponding neighbourhood of type B nodes. 

For the second case the maximum length of the maximal path of degree 2 nodes in $G_2$ is 2 since one end of the maximal path is connected to a heavy node in $G_2$ by an edge from $E_2$ and since there can be at most 1 edge from $E_2$ and at most 2 consecutive edges from $E_1$ in a maximal path of degree 2 nodes in $G_2$ containing edges from $E_1$ and $E_2$. The minimum average density of the nodes of this path is $\frac{1}{2}(2+\frac{1}{2})=\frac{5}{4}>\frac{9}{8}$. Also this path  will not contribute to reduce the density of its corresponding neighbourhood of type B nodes to lower than $\frac{9}{8}$. 

%since its one end is connected to a heavy node in $G_2$ by an edge from $E_2$ and since there can be at most 2 consecutive edges from $E_1$ in a maximal path of degree 2 nodes in $G_2$

%For a maximal path with ends connected to one heavy node of type A and another heavy node of type B the maximum length is 2 since one end has an edge from $E_2$. The minimum contribution from such paths towards the node of type B will be $\frac{2 + 1/2}{2} = 5/4 > 9/8$. Also, this path will not contribute to reduce the density to lower than 9/8. 
So, we consider the heavy nodes of this group each of which has exactly 3 paths of degree 2 nodes in $G_2$.
%If a maximal path of degree 2 nodes, say, $p_0p_1p_2p_3$ is attached to a heavy node $p_0$ of type B and the other end is attached to a heavy node $p_3$ of type A then in the second round the adversary can fix the length of $p_2p_3$ as $|p_2p_3| = |p_0p_1|$. Then $p_1$ and $p_2$ will be ambiguous. There must be an adge at $p_1$ or $p_2$ to make them unambiguous. Remaining degree 2 path will have both ends attached to heavy nodes of type B. So, for this group if the other end of a degree 2 path is connected to a heavy node of type A its maximum length will be 1.
%We consider average densities of nodes of this group for different combinations of 3 degree 2 paths. 
If the group of nodes around a heavy node of type b has 2 degree 2 paths of length 3 attached to heavy nodes of type B and 1 path of degree 2 nodes attached to a heavy node of type A by an edge from $E_2$ then each of the 3 paths will have an edge from $E_2$. In a way similar to the case of group a nodes consisting of 3 paths of degree 2 nodes (Fig.~\ref{333group:label}) it can be shown that the reduced group will have density of at least $\frac{9}{8}$.
%If 2 paths are of length 3 with the other end attached to type B nodes and 1 path is of length 1 with the other end attached to type A node then in a procedure similar to the case of 3 paths of lenght 3 for group a above we can show that to make the placement unique there must be at least one edge incident on a node of the paths such that the reduced group will have density of at least $\frac{9}{8}$. 
%For the group with one path of length 3 and another path of length 2 being attached to heavy node of type B and the third path of length 1 being attached to heavy node of type A, since there is at least one edge from $E_2$ in each path they can be made ambiguous in the second round. The algorithm must have to introduce at least one extra edge at a node of the paths to make the placement unique. Then the reduced group will have average density of at least $\frac{9}{8}$. 
%For the group with one path of length 3 and another path of length 2 being attached to heavy node of type B and the third path of length 1 being attached to heavy node of type A, since there is at least one edge from $E_2$ in each path they can be made ambiguous in the second round. The algorithm must have to introduce at least one extra edge at a node of the paths to make the placement unique. Then the reduced group will have average density of at least $\frac{9}{8}$.
For the group with 2 paths of length 3 being attached to heavy node of type B and the third path of length 0 being attached to node of type A by an edge from $E_1$ the average density is $\frac{1}{4}(4+\frac{1}{2}) = \frac{9}{8}$.
%Similarly, we can show that for the group with one path of length 3 being attached to heavy node of type B and 2 paths of length 1 being attached to heavy nodes of type A the reduced group will have minimum average density of 9/8. 
It can be easily checked that for all other combinations of the maximal paths the minimum average density for the groups of nodes will be at least $\frac{9}{8}$. 

%First we consider the heavy node which has 2 paths of degree 2 nodes of length 2 attached to heavy nodes of type B and one maximal path of degree 2 nodes of length 2 that is attached to a heavy node of type A (Fig.~\ref{222Agroupb:label}). Since the path attached to the node of type A is of length 2 the edge incident on type A node will be in $E_2$. In the second round the adversary can set $|p_2p_3| = |p_0p_1|$ and $|p_0p_3| = |p_1p_2|$. Then the points $p_1$ and $p_2$ will be ambiguous. The algorithm must introduce an edge at $p_1$ or $p_2$ to make them unique. One half of this edge is counted towards the incident node of type A and the other half is counted towards the incident node of type B. Average density of the group is $(3 + 2.5 + 0.5)/(3 + 2) = 6/5$. 

%%\begin{figure}[!h] 
%\centering 
%\includegraphics[scale=1.0]{222Agroupb.pdf} 
%\caption{Heavy node of group b with 2 paths of degree 2 nodes of length 3 and 1 path of degree 2 node of length 1.} 
%\label{222Agroupb:label} 
%\end{figure} 

\end{enumerate} 
\end{enumerate}

Thus, the minimum average density for all nodes in $G_2$ will be $\frac{9}{8}$. \qed 
\end{proof}

\begin{comment}
\section{Conclusions} 

It is challenging to decrease the gap between the upper and lower bounds for two rounds even further. Improving the upper bound of $5n/4$ for three rounds~\cite{DBLP:conf/wabi/ChinLSY07} can also be investigated.
Another interesting line of work is to generalize this problem to learning a set of points in the plane. 
\end{comment}

%\section{Acknowledgements}

%Research supported by an NSERC discovery grant awarded to the second author.

%\bibliographystyle{agsm}    % or some other suitable package.
\bibliographystyle{abbrv}
%\bibliography{CRPITExample}  % often included from a separate file.
\bibliography{pointPlacement}

\newpage
\appendix
\section{}

\begin{lemma} \label{repconp1p2neq2r2lemma:label}  
The 7-cycle $(p_1, q_1, r_1, s, r_2, q_2, p_2)$ of the 3-path basic component of Fig.~\ref{3pathfig:label} cannot be drawn as the layer graph 
of Fig.~\ref{configa:label} if the edges of the component satisfy the following conditions:  

\begin{multline}\label{eqn:repconlemma}
\{|p_1p_3| \ne ||p_3q_3| \pm |r_3s||, |p_1p_3| \ne |r_3s||, ||p_3q_3| \pm |sr_2||\ne |p_2q_2|,\\
||p_3q_3| \pm |sr_2| \pm |sr_3||\ne |p_2q_2|\}  
\end{multline}

%\begin{equation}{\label{eqn:repconlemma}}  
%\{|p_1p_3| \ne ||p_3q_3| \pm |r_3s||, |p_1p_3| \ne |r_3s||, ||p_3q_3| \pm |sr_2||\ne |p_2q_2|, ||p_3q_3| \pm |sr_2| \pm |sr_3||\ne |p_2q_2|\}  
%\end{equation}  

%\begin{center}
%$|p_1p_3| \ne ||p_3q_3| \pm |r_3s||, |p_1p_3| \ne |r_3s||, ||p_3q_3| \pm |sr_2||\ne |p_2q_2|,$

%$||p_3q_3| \pm |sr_2| \pm |sr_3||\ne |p_2q_2|$
%\end{center}
\end{lemma}  

\begin{proof}
We prove by contradiction. Assume that the edges of the component (Fig.~\ref{3pathfig:label}) satisfy (\ref{eqn:repconlemma}) but the 7-cycle $(p_1, q_1, r_1, s, r_2, q_2, p_2)$ can be drawn as a layer graph as in Fig.~\ref{configa:label}. Then by Theorem 1 it is not line rigid. Clearly, a set of points cannot be line rigid if any non-empty subset is not line rigid and again by Theorem 1 must have a layer graph representation. Thus, the whole component must have a layer graph drawing. \\  
  
All possible layer graph drawings  of the component in which the layer graph of Fig.~\ref{configa:label} is embedded are as in Fig.~\ref{configaall4fig:label}. This implies that if the 7-cycle $(p_1, q_1, r_1, s, r_2, q_2, p_2)$ of the component of Fig.~\ref{3pathfig:label} has a layer graph representation as in Fig.~\ref{configa:label} then the whole component must have at least one of the 4 layer graph representations as shown in Fig.~\ref{configaall4fig:label}. \\ 
  
Without loss of generality we assume that the component has a layer graph representation of Fig.~\ref{configaall4fig:label}(a). In the present configuration of the layer graph of the component $p_1, q_1, r_1$ and $s$ are on a line which is parallel to $p_2q_2$ and $q_3r_3$. So, we must have $|p_1p_3| = ||p_3q_3| \pm |r_3s||$.\\    
 
This contradicts the first inequality of (\ref{eqn:repconlemma}) which corresponds to the layer graph of the component in the present  
configuration. Hence, the whole component cannot be drawn as a layer graph when its 7-cycle $(p_1, q_1, r_1, s, r_2, q_2, p_2)$ has a layer graph representation in the configuration of Fig.~\ref{configa:label}. In other words, in any layer graph representation (if any one is possible) of the whole component its 7-cycle $(p_1, q_1, r_1, s, r_2, q_2, p_2)$ cannot have a layer graph representation in the configuration of Fig.~\ref{configa:label}. Consequently, the 7-cycle $(p_1, q_1, r_1, s, r_2, q_2, p_2)$ of the component cannot be drawn as a layer graph in the configuration of Fig.~\ref{configa:label}. 
\qed
\end{proof}
%%%%%%%%%%%%%%%%%%%%%%%%%%%%%%%%%%%%%%%%%%%%%

\newpage
%\appendix
\section{}

\begin{lemma} \label{fixedpointlemma:label}
%35 fixed points are sufficient to attach the basic components evenly to them. 
A set $S$ of 35 points is sufficient to ensure that the valences of any two points in $S$ differ by at most 2. 
\end{lemma} 

\begin{proof}
%To attach a basic component we always select a point in $S$ with the lowest valence as the first point (say $p_1$). Of the rest 34 points of $S$, at most 8 points may not be acceptable for the second point (say $p_2$), because of the conditions on $p_1p_2$. From among the rest points that satisfy the conditions on $p_1p_2$ we select the one that has the lowest valence, as $p_2$. Of the rest 33 points of $S$, at most 16 may not be acceptable for the last point, say $p_3$, because of the conditions on $p_2p_3$ and $p_3p_1$. From among the rest points that satisfy the conditions on $p_2p_3$ and $p_3p_1$ we choose the one that has the lowest valence, as $p_3$. This will be follwed to attach each basic component to the points in $S$. We attach the components sequentially.

Initially, we have $|S_0| = 35$ and all other $S_i$'s are of size 0. After attaching the first 6 basic components we have $|S_0| = 17$ and $|S_1| = 18$. Now we attach components until $|S_0| \leq 9$. This will attach at most 4 components. We have $|S_0| \leq 9$ and $|S_2| \leq 4$. The rest points are of valence 1, i.e., $|S_1| \geq 22$.

\begin{comment}
Initially, we have $|S_0| = 35$ and all other $S_i$s are of size 0. We attach the first 6 basic components to 18 points in $S_0$. Then each of those 18 points will have valence exactly one, and remaining 17 points will have valence 0, i.e., $|S_1| = 18$ and $|S_0| = 17$. %Now we repeatedly attach components until there are at most 9 points left in $S_0$.
Now we attach components until $|S_0| \leq 9$. This will attach at least 4 components. We have $|S_2| \leq 4$ and  $|S_0| \leq 9$. The rest points are of valence 1, i.e., $|S_1| \geq 22$.
\end{comment}

%Now we choose at least 2 points $p_1$ and $p_2$ from $S_0$. until there are at most 9 points left in it. Their valences will be raised to 1. If the the third point $p_3$ is not found in $S_0$ we choose it from $S_1$. Its valence will become 2. Thus, we have $|S_2| \leq 4$ and  $|S_0| \leq 9$. The rest points are of valence 1, i.e., $|S_1| \geq 22$. 
%there will be at most 4 points of valence 2, at most 9 points of valence 0 and the rest, i.e., at least 22 points, are of valence 1. 

%Next we repeatedly attach components until there are at most 18 points in $S_0 \cup S_1$. 
Next we attach components until $|S_0| \leq 2$. This will attach at most 7 components. Then we have $|S_0| \leq 2$ and  $|S_2| \leq 18$, and consequently, $|S_1| \geq 15$. Again, we attach components until $|S_0| = 0$. This will attach at most 2 components. Then we have $|S_0| = 0$ and $|S_3| \leq 2$, and consequently,  $|S_1 \cup S_2| \geq 33$. Thus, all the points in $S$ have at most 3 consecutive valences, viz., 1, 2 and 3. At any point of time they may have at most 4 consecutive valences, viz., 0-3. 

%Now we find 3 points in such a way that at least 1 point is of valence 0 and the rest, i.e., at most 2, are of valence 1 until there are at most 18 points of valence 0 and 1. Then there will be at most 2 points of valence 0 and at least 15 points of valence 1. To attach another basic component we can choose at least 2 points from valence 0 and 1 with at least 1 point from valence 0. We choose the remaining at most 1 point from valence 2 points. If chosen its valence will become 3. Thus the 35 fixed points will have 4 consecutive levels of valence. 

%Following above construction, we can use all the points of valence 0 to attach the basic components. Then there will be at least 33 points of valence 1 and 2, and at most 2 points of valence 3. All the points are in 3 consecutive levels of valence.

We shall show that at any point of time the points in $S$ will have at most 4 consecutive valences, and that at some point of time they will have at most 3 consecutive valences only. For this we use induction to show that if we start with points in $S$ in 3 consecutiv valences $d$, $d+1$ and $d+2$, and attach the basic components according to our algorithm, then at some point of time they will have the next 3 valences $d+1$, $d+2$ and $d+3$ only. We assume that $|S_{d} \cup S_{d+1}| \leq 18$. Otherwise, we attach components until $|S_{d} \cup S_{d+1}| \leq 18$.

First, we consider the cases for which $|S_d| \leq 9$. Then $|S_{d+1} \cup S_{d+2}| \geq 26$ with $|S_d \cup S_{d+1} \cup S_{d+2}| = 35$. We attach components until $|S_d| = 0$. For each new component, at least 1 point of $S_d$ will be moved to $S_{d+1}$, and at most 2 points of $S_{d+2}$ will be moved to $S_{d+3}$. It is clear that at most 9 components will be attached, and that there will always be at least 19 points in $S_d \cup S_{d+1} \cup S_{d+2}$ until there is no point in $S_d$. We have $|S_d| = 0$, $|S_{d+1} \cup S_{d+2}| \geq 17$ and $|S_{d+3}| \leq 18$. Thus, the valences of all the points will become $d+1$, $d+2$ and $d+3$. 

%Now we consider another worst case when there are 9 points of valence $d$, 0 point of valence $d+1$ and the remaining 26 points are of valence $d+2$. For each basic component we can find at least 1 point in valence $d$ and the rest from valences $d+1$ and $d+2$. It is evident that there will always be at least 19 points in valences $d$, $d+1$ and $d+2$ until there is no point of valence $d$. Thus, valences of all the points will become $d+1$, $d+2$ and $d+3$. 

Now we consider the worst case for which $|S_d| = 18$ and $|S_{d+1}| = 0$. They imply that $|S_{d+2}| = 17$. We attach components until $|S_d| \leq 10$. At most 4 components will be attached. We group all the possible situations into 2 subcases. First, we consider the subcase when 2 points are used from $S_d$ for each new component. Exactly 4 components will be attached using 8 points from $S_d$. We have $|S_d| =10$ and  $|S_{d+1}| \geq 5$ with $|S_d \cup S_{d+1}| \geq 15$, and $|S_{d+3}| \leq 4$. After attachment of 1 more component we have $|S_d| \leq 8$ and $|S_{d+1}| \geq 6$ with $|S_d \cup S_{d+1}| \geq 14$, and  $|S_{d+3}| \leq 5$. Now we attach components until $|S_d| \leq 5$. Clearly, at most 3 components will be attached, and we have $|S_d| \leq 5$ and $|S_{d+1}| \geq 3$ with $|S_d \cup S_{d+1}| \geq 8$ (because at most 6 valence $d+1$ points will be raised to valence $d+2$ points), and  $|S_{d+3}| \leq 8$ (because at most 3 valence $d+2$ points will be raised to valence $d+3$ points). As long as there are at least 19 points in $S_d \cup S_{d+1} \cup S_{d+2}$, all the 3 points of a new component will be chosen from that union. No points will be used from $S_{d+3}$, and hence no point's valence will be raised to $d+4$. We attach components until $|S_d| = 0$. It is evident that at most 5 components will be attached, and we have $|S_d| = 0$, $|S_{d+1} \cup S_{d+2}| \geq 17$ and $|S_{d+3}| \leq 18$.  

Now we consider the other subcase which consists of the remaining possible situations. For this case, 3 or 4 components will be attached. It can be easily seen that $|S_d| \leq 9$ and  $|S_{d+1}| \geq 6$ with $|S_d \cup S_{d+1}| \geq 15$, and $|S_{d+3}| \leq 3$. We attach compnents until $|S_d| \leq 6$. It can be easily checked that at most 3 components will be attached, and we have $|S_d|\leq 6$ and  $|S_{d+1}| \geq 3$ with $|S_d \cup S_{d+1}| \geq 9$, and $|S_{d+3}| \leq 6$. We attach components until $|S_d| = 0$.It is evident that at most 6 components will be attached, and we have $|S_d| = 0$, $|S_{d+1} \cup S_{d+2}| \geq 17$ and $|S_{d+3}| \leq 18$.  

%Then there will be at most 4 points of valence $d+3$ and at least 13 points of valence $d+2$. For each point of valence $d+3$ there will be at least 2 points of valence $d+1$. Now we can find at least 1 point of valence $d$ until there is no point of valence $d$ subject to at least 2 points from valences $d$ and $d+1$ until there are at least 10 points in valences $d$ and $d+1$. It can be easily seen that there will be at least 19 points in valences $d$, $d+1$ and $d+2$ until there is no point of valence $d$. Thus, all the points will be elevated to valences $d+1$, $d+2$ and $d+3$.

%To elevate the valences to $d+1$, $d+2$ and $d+3$ we need to use all the points of valence $d$ to attach the basic components. Let us consider the worst case when there are 18 points in valence $d$ and no points in valence $d+1$. We can find at least 2 points from valence $d$ for each component until there are at most 10 points of valence $d$. Then there will be at most 4 points of valence $d+3$ and at least 13 points of valence $d+2$. For each point of valence $d+3$ there will be at least 2 points of valence $d+1$. Now we can find at least 1 point of valence $d$ until there is no point of valence $d$ subject to at least 2 points from valences $d$ and $d+1$ until there are at least 10 points in valences $d$ and $d+1$. It can be easily seen that there will be at least 19 points in valences $d$, $d+1$ and $d+2$ until there is no point of valence $d$. Thus, all the points will be elevated to valences $d+1$, $d+2$ and $d+3$.

It can be easily shown that for all the other combinations of number of points in valences $d$ and $d+1$ subject to a maximum of 18, all the points will be elevated to at most 3 consecutive valences $d+1$, $d+2$ and $d+3$. The calculations will be similar to the above.\qed

%It is evident that all the other combinations of number of points in valences $d$ and $d+1$ subject to a maximum of 18 will have situations similar to the above 2 constructions.

%Thus, in all cases all the points will be elevated to at most 3 consecutive valences $d+1$, $d+2$ and $d+3$. \qed
\end{proof}

%%%%%%%%%%%%%%%%%%%%%%%%%%%%%%%%%%%%%%%%%%%%%
\newpage
%\appendix
\section{}

\begin{lemma} \label{maxpathlemma:label}
The number of nodes in any maximal path of degree 2 nodes in $G_2$ is at most 3. 
\end{lemma} 

\begin{proof}

If a maximal path of degree 2 nodes of $G_2$ consists of edges from $E_1$ only then by Step 3 of adversary its lenght is at most 3.

Now we consider maximal path of degree 2 nodes of $G_2$ that contains at least one edge from $E_2$. In such a path
%In a maximal path of degree 2 nodes in $G_2$,
 there cananot be three consecutive edges from $E_1$ because of $S_3$. Suppose the number of degree 2 nodes in a maximal path is 4. Let the nodes be $p_1, p_2, p_3$ and $p_4$. Let $p_0$ and $p_5$ be heavy nodes adjacent to $p_1$ and $p_4$ respectively. Since any maximal path of degree 2 in $G_2$ can have at most 2 consecutive edges from $E_1$ we can have the following 5 combinations of the $E_1$ and $E_2$ type edges for the edges $p_0p_1$, $p_1p_2$, $p_2p_3$, $p_3p_4$ and $p_4p_5$:\\

%\noindent 
%$E_2$, $E_1$, $E_2$, $E_1$, $E_1$ or $E_2$, $E_1$, $E_1$, $E_2$, $E_1$ or $E_1$, $E_2$, $E_1$, $E_2$, $E_1$ or 
%$E_1$, $E_1$, $E_2$, $E_2$, $E_1$ or $E_1$, $E_1$, $E_2$, $E_1$, $E_1$. \\ 

\begin{enumerate}   
\item $E_2$, $E_1$, $E_2$, $E_1$, $E_1$   
\item $E_2$, $E_1$, $E_1$, $E_2$, $E_1$   
\item $E_1$, $E_2$, $E_1$, $E_2$, $E_1$   
\item $E_1$, $E_1$, $E_2$, $E_2$, $E_1$   
\item $E_1$, $E_1$, $E_2$, $E_1$, $E_1$   
\end{enumerate}   

For combination 1, since there are two edges in $E_2$ lengths of those edges can be set in such a way that  
$|p_0p_5| = |p_1p_2| + |p_2p_3|$ and $|p_0p_1| = |p_4p_5| - |p_3p_4|$, and the graph $G_2$ becomes non-rigid (Fig.~\ref{deg2maxpath-case1label}). 

\begin{figure}[ht]   
\begin{minipage}[b]{0.5\linewidth}   
\centering   
\includegraphics[scale=0.7]{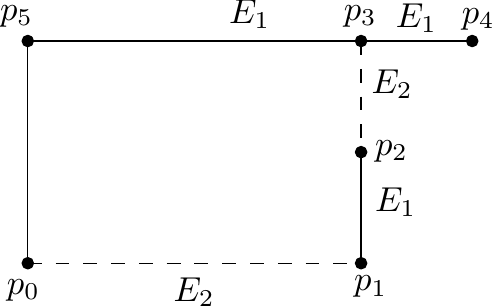}   
\caption{Maximal path of degree 2 in $G_2$ for the combination of edges $E_2, E_1, E_2, E_1, E_1$.}   
\label{deg2maxpath-case1label}   
\end{minipage}   
\hspace{0.2cm}   
\begin{minipage}[b]{0.45\linewidth}   
\centering   
\includegraphics[scale=0.7]{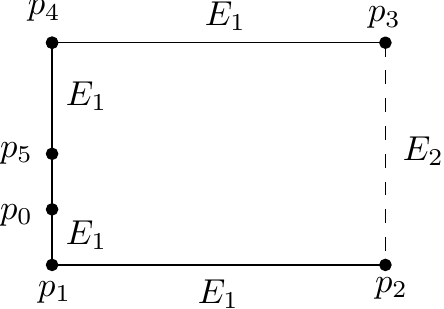}   
\caption{Maximal path of degree 2 in $G_2$ for the combination of edges $E_1, E_1, E_2, E_1, E_1$.}   
\label{deg2maxpath-case5label}   
\end{minipage}   
\end{figure} 

Similarly, for combinations 2-4 the adversary can make the graph ambiguous. As for combination 5, the adversary can set  
$|p_1p_2| = |p_3p_4| = c$ in the first round by $S_2$ and can set the length of $p_2p_3$ in round 2 in such a way  
that $|p_2p_3| = |p_4p_5| + |p_5p_0| + |p_0p_1|$ (Fig.~\ref{deg2maxpath-case5label}). Then the cycle $(p_0, p_1, p_2, p_3, p_4, p_5)$ will not be line rigid .\qed

\end{proof} 

%%%%%%%%%%%%%%%%%%%%%%%%%%%%%%%%%%%%

\begin{comment}

\end{comment}

\end{document}